\documentclass[english, 12pt]{article}
\usepackage{amsmath,amssymb,amsthm}
\usepackage{babel}
\usepackage{mathtools}
\usepackage{natbib}
\usepackage[nodisplayskipstretch]{setspace}
\usepackage{bm}
\usepackage{threeparttable}
\usepackage{graphicx}
\usepackage{enumerate}
\usepackage{commath}			% for \norm{...}
\usepackage{bibentry}
\usepackage{booktabs}
\usepackage[titletoc,title]{appendix}
\usepackage{multirow}
\usepackage{arydshln}
\usepackage{tabularx}
\usepackage{dsfont}
\usepackage{placeins}
\usepackage{rotating} % Rotating table
\usepackage[font={small}]{subcaption}
\usepackage[left=1in, right=1in, top=1in, bottom=1in]{geometry}

\usepackage[colorinlistoftodos,textsize=small]{todonotes}

\allowdisplaybreaks

\usepackage{url}
\usepackage{xr-hyper}  % for cross-referencing
\usepackage[
	breaklinks,
	colorlinks=true,
	pagebackref=false,
	pdfauthor={},%
	pdftitle={},% 
	citecolor=blue,
	linkcolor=blue,
	urlcolor=blue
	]{hyperref}      

\input{ee.sty}

\newtheorem{prop}{Proposition}

\newtheorem{thm}{Theorem}
\newtheorem{algorithm}{Algorithm}

\newtheoremstyle{remark2}{1ex}{1ex}%
      {}% Body font
      {}% Indent amount (empty = no indent, \parindent = para indent)
      {\bf}% Thm head font. notice \YYshape command works unlike \textyy
      {.}% Punctuation after thm head
      {5pt}% Space after thm head (\newline = linebreak)
      {\thmname{#1}\thmnumber{ #2}\thmnote{ \slshape{(#3)}}} % Thm head spec, so numbering first
\theoremstyle{remark2}
\newtheorem{rem}{Remark}

\newtheoremstyle{remark3}{1ex}{1ex}%
      {}% Body font
      {}% Indent amount (empty = no indent, \parindent = para indent)
      {\bf}% Thm head font. notice \YYshape command works unlike \textyy
      {.}% Punctuation after thm head
      {5pt}% Space after thm head (\newline = linebreak)
      {\thmname{#1}\thmnumber{.#2}\thmnote{ \slshape{(#3)}}} % Thm head spec, so numbering first
\theoremstyle{remark3}

\newcommand{\1}{\mathds{1}}
\makeatletter
\renewenvironment{proof}[1][\bfseries\proofname]{\par
   \pushQED{\qed}%
   \normalfont \topsep6\p@\@plus6\p@\relax
   \trivlist
   \item[\hskip\labelsep
     #1\@addpunct{:}]\ignorespaces
}{%
   \popQED\endtrivlist\@endpefalse
}
\makeatother

\newcommand{\Comments}{1}
\newcommand{\mynote}[2]{\ifnum\Comments=1\textcolor{#1}{#2}\fi}
\newcommand{\mytodo}[2]{\ifnum\Comments=1%
  \todo[linecolor=#1!80!black,backgroundcolor=#1,bordercolor=#1!80!black]{#2}\fi}

\ifnum\Comments=1               % fix margins for todonotes
\fi

\ifnum\Comments=1               % fix margins for todonotes
\fi

\newcommand{\CoVaR}{\operatorname{CoVaR}}
\newcommand{\RCoVaR}{\operatorname{RCoVaR}}
\newcommand{\CoES}{\operatorname{CoES}}
\newcommand{\MES}{\operatorname{MES}}

\newcommand{\VaR}{\operatorname{VaR}}

\newcommand{\Bern}{\operatorname{Bern}}

\newcommand{\D}{\,\mathrm{d}}

\renewcommand{\E}{\mathbb{E}}
\renewcommand{\P}{\mathbb{P}}
\renewcommand{\a}{\alpha}
\renewcommand{\b}{\beta}

\begin{document}

\baselineskip18pt
% \addtolength{\floatsep}{-5cm}
% \renewcommand\floatpagefraction{.8}
% \renewcommand\textfraction{.05}
\renewcommand\floatpagefraction{.9}
\renewcommand\topfraction{.9}
\renewcommand\bottomfraction{.9}
\renewcommand\textfraction{.1}
\setcounter{totalnumber}{50}
\setcounter{topnumber}{50}
\setcounter{bottomnumber}{50}
\abovedisplayskip1.5ex plus1ex minus1ex
\belowdisplayskip1.5ex plus1ex minus1ex
\abovedisplayshortskip1.5ex plus1ex minus1ex
\belowdisplayshortskip1.5ex plus1ex minus1ex

\title{Systemic Risk Surveillance\thanks{
		The first author gratefully acknowledges support of the Deutsche Forschungsgemeinschaft (DFG, German Research Foundation) through grant 502572912 and the second author through grants 460479886, 531866675 and 568876076. 
		Replication material for the simulations and application is available on Github under \href{https://github.com/TimoDimi/replication_CoVaR_Monitoring}{https://github.com/TimoDimi/replication\_CoVaR\_Monitoring}.
	}
}

%\begin{notthisone}
\author{
	Timo Dimitriadis\thanks{Faculty of Economics and Business, Goethe University Frankfurt, 60629 Frankfurt am Main, Germany, and Heidelberg Institute for Theoretical Studies, \href{mailto:dimitriadis@econ.uni-frankfurt.de}{dimitriadis@econ.uni-frankfurt.de}.}
\and 
	Yannick Hoga\thanks{Faculty of Economics and Business Administration, University of Duisburg-Essen, Universit\"atsstra\ss e 12, D--45117 Essen, Germany, \href{mailto:yannick.hoga@vwl.uni-due.de}{yannick.hoga@vwl.uni-due.de}.}
}

\date{\today}
%\end{notthisone}
\maketitle

\begin{abstract}
	\noindent 
	Following several episodes of financial market turmoil in recent decades, changes in systemic risk have drawn growing attention.
	Therefore, we propose surveillance schemes for systemic risk, which allow to detect misspecified systemic risk forecasts in an ``online'' fashion.
	This enables daily monitoring of the forecasts while controlling for the accumulation of false test rejections.
	Such online schemes are vital in taking timely countermeasures to avoid financial distress. 
	Our monitoring procedures allow multiple series at once to be monitored, thus increasing the likelihood and the speed at which early signs of trouble may be picked up. 
	The tests hold size by construction, such that the null of correct systemic risk assessments is only rejected during the monitoring period with (at most) a pre-specified probability. 
	Monte Carlo simulations illustrate the good finite-sample properties of our procedures.	An empirical application to US banks during multiple crises demonstrates the usefulness of our surveillance schemes for both regulators and financial institutions. \\
	
	\noindent \textbf{Keywords:} CoVaR, Forecasting, Monitoring, Multiple Testing, Systemic Risk \\
	\noindent \textbf{JEL classification:} C52 (Model Evaluation, Validation, and Selection); G17 (Financial Forecasting and Simulation); G32 (Financial Risk and Risk Management)
\end{abstract}

% \thispagestyle{empty}
% \clearpage
% \addtocounter{page}{-1}

\section{Motivation}

% \onehalfspacing      % spacing
\doublespacing      % spacing

The numerous financial crises of recent times and their severe economic reverberations have raised awareness of the importance of systemic risk \citep{AB16,Aea17,VZ19}. 
To better appreciate the difference between risk and \textit{systemic} risk, consider the returns on shares of banks.
While bank returns are often subject to higher volatility during crises (implying larger risk), this does not necessarily translate into increased commonality (implying larger \emph{systemic} risk). 
However, it is precisely the increased commonality that regulators have come to be most concerned about, as this may entail system-wide distress with potentially severe economic costs \citep{GKP16}.

Therefore, it is important to evaluate systemic risk forecasts in the financial system.
While classical financial ``one-shot'' backtests are designed for such forecast evaluations, under repeated application their statistical type I errors (false test rejections) accumulate over time---up to the point that a true null will be rejected with probability approaching one.
As systemic risks are to be monitored continuously (e.g., daily), it is necessary to control \emph{accumulated} false rejections, aligning with recent interest in safe anytime-valid statistical inference \citep{shafer2021testing, vovk2021values, ramdas2023game, WWZ23}.
In this paper, we interchangeably call  methods with ``time-uniform'' false rejection guarantees monitoring procedures, surveillance schemes or online tests.

Despite the apparent need for systemic risk surveillance, there is a lack of statistically valid tools for this in the literature. 
It is the main aim of this paper to fill this gap by providing such  tools and to show their validity.
To the best of our knowledge, there only exist ``one-shot'' backtests that assess the adequacy of systemic risk forecasts \citep{Bea21,FH21}.
However, their rejection rates would accumulate under repeated application.
We propose monitoring procedures for the most popular systemic risk measure, the conditional Value-at-Risk (CoVaR) of \citet{AB16}, and the reverse CoVaR (RCoVaR).
In Appendix~\ref{sex:CoESandMES}, we also develop monitoring schemes for the conditional expected shortfall (CoES) and marginal expected shortfall (MES) of  \citet{Aea17}.

For the construction of the monitoring procedure for the CoVaR as the most important systemic risk measure, we draw on recent work of \citet{FH21}.
They provide so-called {identification functions} that uniquely identify both the stand-alone risk and systemic risk.
Under the null hypothesis that the CoVaR forecasts are correctly specified (i.e., calibrated), the probabilistic structure of the corresponding identification functions is fully known: they are independent and identically distributed (IID) Bernoulli random variables. This remains true without making any additional assumptions about the dynamic structure of the return series.
As the probabilistic structure under the null is fully known, it can be used to simulate time-uniform and non-asymptotic critical values.
The simulation-based construction is particularly useful as systemic risk events are rare by definition, such that standard asymptotic theory may not be relied upon to provide accurate approximations in finite samples \citep{Hog19a+}.
We also extend our monitoring procedure to the CoVaR with reversed conditioning.

Because the probabilistic structure of the identification functions for CoES and MES is not fully characterized under the null, a direct extension of our CoVaR monitoring scheme is infeasible. Instead, in Appendix~\ref{sex:CoESandMES}, we exploit a systemic-risk analogue of the classical \textit{cumulative violation sequence} used in \citet{Acerbi2002spectral}, \citet{DE17}, \citet{Du2024powerful}, among others. 
We show that, under the null, the full probabilistic structure of this sequence---namely its uniform-type distribution and independence---is known, which again enables the computation of non-asymptotic, time-uniform critical values via simulation.

Our approach of constructing the CoVaR monitoring procedure is closely related to \citet{HD22a+}, who propose online detection schemes for single \textit{univariate} risk forecasts, viz.~Value-at-Risk (VaR) and expected shortfall (ES). 
However, we improve upon their procedure by using normalized detectors, and most importantly, extend their theory to \emph{systemic} risk measures.
The latter necessitates monitoring multiple time series at once, whereas \citet{HD22a+} only consider a single series.
This multivariate monitoring necessitates a Bonferroni-type correction to obtain critical values with finite-sample validity, which are not too conservative in practice.
A further benefit of the Bonferroni correction is that it allows to attribute a monitoring alarm to a specific financial institution.

Additional work that is related to ours are the monitoring procedures proposed by \citet{WG13}, \citet{NLL14} and \citet{MSW20}. 
These authors do, however, not focus on systemic risk itself, but only on quantities that can at best be described as crude approximations to it, such as correlation \citep{WG13} and the copula \citep{NLL14,MSW20}. 
Moreover, these procedures rely on an initial period of non-contamination, which our approach does not require.

Regarding calibration (backtesting) under repeated use, our work relates to the literature on safe and anytime-valid inference (SAVI) based on e-values. For VaR and ES, several e-value--based monitoring schemes have been proposed, but they do not extend to \emph{systemic} risk measures. While e-value tests are typically valid indefinitely, they are often overly conservative. In contrast, our approach fixes the monitoring horizon and delivers non-asymptotic, simulation-based, and hence (almost) exact size control via (almost) exact critical values.
\citet{horvath2025sequential} take yet another approach by proposing a sequential monitoring procedure for VaR and ES models that is, however, highly model-specific and based on asymptotic theory.

Our Monte Carlo simulations use a realistic DCC--GARCH model \citep{Eng02} for the asset returns to demonstrate the good finite-sample properties of our  monitoring procedures. 
Despite the use of a potentially conservative Bonferroni-type correction, the empirical size remains close to the nominal level---even in moderately high-dimensional settings, where numerous individual tests could, in principle, amplify such conservative tendencies.
We also show that, under the null, the detector to first raise a false alarm pertains to the VaR or the CoVaR with equal probability. The simulations further demonstrate the high power of our procedures in identifying misspecified forecasts, again exhibiting a balanced pattern regarding which detector first signals a suboptimal predictive performance.

In the empirical application, we apply our monitoring procedure to the stock returns of systemically relevant US banks.
In doing so, we evaluate systemic risk forecasts from two competing multivariate GARCH models over a calm period and three turbulent episodes: the global financial crisis, the COVID-19 pandemic, and the recent US~tariff period under President Trump.
The monitoring results indicate that the constant covariance structure of the simpler CCC–GARCH model is inadequate for capturing the \emph{systemic} component of financial risk during such volatile periods, whereas the forecasts of the more flexible DCC–GARCH cannot be rejected.

The remainder of the paper is structured as follows. 
Section~\ref{Systemic Risk Monitoring} formally introduces the CoVaR and its reverse variant and our corresponding monitoring procedures. 
A related monitoring scheme for the CoES and MES is constructed in Appendix~\ref{sex:CoESandMES}.
Simulations in Section~\ref{Simulations} explore the finite-sample performance of our methods, and Section~\ref{Empirical Application} presents the empirical application.
Finally, Section~\ref{Conclusion} concludes. 
Next to our CoES and MES surveillance schemes in Section~\ref{sex:CoESandMES}, the appendix contains Section~\ref{sec:Proofs of the main paper} with the proofs for our CoVaR procedures, and Section~\ref{sec:Algorithms} that contains the algorithms to compute critical values.

\section{Systemic Risk Monitoring}
\label{Systemic Risk Monitoring}

First, we define the systemic risk measures we consider in Section~\ref{Defining Systemic Risk}.
Then, we introduce monitoring procedures for the CoVaR in Section~\ref{CoVaR Monitoring}, which we extend to the reverse CoVaR in Section~\ref{ReversedCoVaRMonitoring}. 
Surveillance schemes for the CoES and MES are presented in Appendix~\ref{sex:CoESandMES}.

\subsection{Defining Systemic Risk Measures}
\label{Defining Systemic Risk}

\sloppy 
Consider the to-be-monitored sequence of random variables $\big\{(X_t,Y_{1t},\ldots,Y_{Kt})^\prime\big\}_{t \in \mathbb{N}}$, where $K \in \mathbb{N}$ is the number of monitored variables.
Here, $Y_{kt}$ ($k=1,\ldots,K$) stands for the log-losses of interest (e.g., the losses of a bank's shares / losses of a business unit) and $X_t$ are the log-losses of some reference position (e.g., system-wide losses in the financial system / bank-wide losses of all business units).
In a concrete monitoring situation, the variables at time $t=1,2\ldots$ are still to be (sequentially) observed and are not yet available when setting up the surveillance scheme at time $t=0$.

Let $\mathcal{F}_{t}=\sigma\big((X_t,Y_{1t},\ldots,Y_{Kt})^\prime,\mZ_t,(X_{t-1},Y_{1,t-1},\ldots,Y_{K,t-1})^\prime,\mZ_{t-1},\ldots\big)$ denote the time-$t$ information set with (possibly multivariate) $\mZ_t$ containing additional covariates.
Further, we denote by $F_{(X_t,Y_{kt})^\prime\mid\mathcal{F}_{t-1}}(\cdot,\cdot)$ the cumulative distribution function (CDF) of $(X_t,Y_{kt})^\prime\mid\mathcal{F}_{t-1}$ for $k\in[K]$, where we use the shorthand notation $[K] := \{1,\dots,K\}$ for any $K \in \mathbb{N}$.
Throughout this paper, we assume that $(X_t,Y_{kt})^\prime\mid\mathcal{F}_{t-1}$ has a strictly positive Lebesgue density for all $(x,y)\in\mathbb{R}^2$ such that $F_{(X_t,Y_{kt})^\prime\mid\mathcal{F}_{t-1}}(x,y)\in(0,1)$.
This assumption ensures that we do not have to rely on generalized inverses to compute our (quantile-based) risk measures, which we introduce next.

Define the VaR as the quantile at level $\b \in (0,1)$ of the conditional distribution $F_{X_t\mid\mathcal{F}_{t-1}}$, i.e., $\VaR_{t,\b}=F_{X_t\mid\mathcal{F}_{t-1}}^{-1}(\b)$.
The VaR is a popular \textit{univariate} risk measure.
Note that the inverse $F_{X_t\mid\mathcal{F}_{t-1}}^{-1}(\cdot)$ exists thanks to our assumption on the distribution of $(X_t, Y_{kt})^\prime\mid\mathcal{F}_{t-1}$.

We now introduce the systemic risk measures CoVaR and its reverse variant that we denote by RCoVaR.
The stress event in the definition of all these measures is that the loss of the reference position exceeds its VaR, i.e., $\{X_t\geq\VaR_{t,\b}\}$.
Then, we define the CoVaR of the $k$-th series $Y_{kt}$ as
\begin{align}
	\label{eqn:DefCoVaR}
	\CoVaR_{kt,\a|\b}=F_{Y_{kt}\mid X_t \geq \VaR_{t,\b},\mathcal{F}_{t-1}}^{-1}(\a),
\end{align}
where $F_{Y_{kt}\mid X_t \geq \VaR_{t,\b}, \mathcal{F}_{t-1}}(\cdot)=\P\{Y_{kt}\leq \cdot\mid X_t \geq \VaR_{t,\b},\ \mathcal{F}_{t-1}\}$, and $\a \in (0,1)$.
As for the VaR, our assumption on the CDF of $(X_t, Y_{kt})^\prime\mid\mathcal{F}_{t-1}$ ensures that the inverse in \eqref{eqn:DefCoVaR} exists.
Note that the CoVaR is simply the $\a$-VaR of the distribution of $Y_{kt}$ \textit{conditional} on $\{X_t\geq\VaR_{t,\b}\}$, i.e., conditional on $X_t$ being \emph{in distress}. 
Without the conditioning on the distress event $\{X_t\geq\VaR_{t,\b}\}$ (or, equivalently, with $\beta=0$), the CoVaR reduces to the VaR of $Y_{kt}$, such that 
$\CoVaR_{t,\a|0} = F_{Y_{kt}\mid\mathcal{F}_{t-1}}^{-1}(\a)$.
Since $X_t$ denotes financial losses, we typically consider values for $\a$ and $\beta$ close to one for the CoVaR to capture (downside) systemic risk (e.g., $\a=\b=0.95$).
When $\a=\b$ we simply write $\CoVaR_{kt,\a} = \CoVaR_{kt,\a|\a}$. 

Following among others \citet{GT13}, \citet{NZ20} and \citet{Bea21}, our CoVaR definition deviates from the original one of \citet{AB16}, who use $\{X_t=\VaR_{t,\b}\}$ as the stress event. This latter choice is problematic because it often has probability zero and it does not fully incorporate all tail events of $X_t$.
We refer to \citet[Section~2.1]{DH24} for a detailed account of the advantages of the definition in \eqref{eqn:DefCoVaR}.

The above interpretation of the CoVaR in \eqref{eqn:DefCoVaR} with $Y_{kt}$ as losses of financial institutions and $X_t$ as market losses coincides with what \citet[Sec.~II.D]{AB16} call the \textit{Exposure CoVaR}, and it measures how institution $k$ is \emph{exposed} to market risk.
A regulator might however also be interested in the reverse direction, i.e., how an individual bank \textit{affects} the market.
Therefore, we additionally define the CoVaR with \emph{reverse} conditioning as 
\begin{align}
	\label{eqn:DefRevCoVaR}
	\RCoVaR_{kt,\a|\b} = F_{X_t\mid Y_{kt} \geq \VaR_{kt,\b},\mathcal{F}_{t-1}}^{-1}(\a),\qquad k \in [K],
\end{align}
where $\VaR_{kt,\b} = F_{Y_{kt}\mid\mathcal{F}_{t-1}}^{-1}(\b)$ is now specific to institution $k$.
The reverse CoVaR in \eqref{eqn:DefRevCoVaR} measures the impact of distress in bank $k$ on the general market, which corresponds to the notion of banks as transmitters of systemic risk.
In contrast, our initial CoVaR definition views banks as receivers of systemic risk.
Of course, both definitions reveal useful information, and which definition is more suitable depends on the context.

\subsection{CoVaR Monitoring}
\label{CoVaR Monitoring}

We now propose procedures for monitoring the correct specification of CoVaR forecasts.
Importantly, the monitoring should be carried out in an \textit{online} fashion, i.e., sequentially as new observations become available.
Classical backtests with asymptotic validity based on a test statistic $Z_T$ with sample size $T$ and critical value $c_\iota$ for significance level $\iota \in (0,1)$ require that under the null hypothesis of correct specification, $H_0$, the asymptotic rejection probability $\lim_{T \to \infty} \P_{H_0}\{Z_T > c_\iota\} \le \iota$ is below $\iota$.
In contrast, sequential \emph{monitoring procedures} satisfy the stronger non-asymptotic notion
\begin{align}
	\label{eqn:MonitoringGuarantee}
	\P_{H_0} \big\{ \exists T \in \mathcal{I}: \; Z_T > c_\iota \big\} \le \iota.
\end{align}
This implies that even though the test decision is monitored every trading day---not unusual in a risk management context---the \textit{accumulated} probability of a  false rejection (type I error) remains controlled.
We will achieve such \textit{finite sample} type I error guarantees in \eqref{eqn:MonitoringGuarantee} by constructing detectors (essentially: test statistics), whose probabilistic structure is (almost) fully known under the null hypothesis, such that exact critical values can be calculated.

\begin{figure}[tb]
	\centering
	\includegraphics[width=\textwidth]{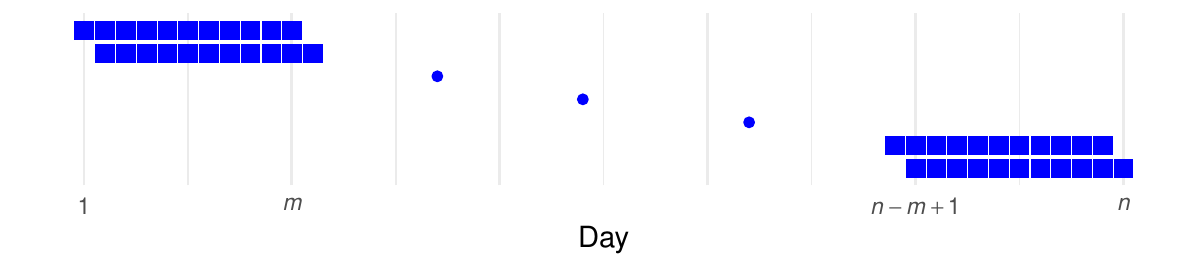}
	\caption{Illustration of the Monitoring Windows.}
	\label{fig:MonitoringWindows}
\end{figure}

In a practical monitoring situation, we are about to sequentially observe $(X_1,Y_{11},\ldots,Y_{K1})^\prime, \ldots, (X_n,Y_{1n},\ldots,Y_{Kn})^\prime$ at times (trading days) $1,\dots,n$.
Our monitoring procedure defined below requires a \emph{rolling monitoring window} of fixed length $m \le n$, as illustrated by the blue bars in Figure~\ref{fig:MonitoringWindows}.
Thus, at each \textit{monitoring time} $T \in \{m,\dots,n\}$---the endpoint of a given monitoring window---the procedure relies on the window spanning the times $\{T-m +1 ,\dots,T\}$.
We only start monitoring when the first full window of $m$ trading days is available at time $T=m$.
While starting earlier with an initially expanding window would be possible, this would come at the cost of an explosion of technicalities, which we avoid for better readability.

We now formalize the null hypothesis of \emph{risk measure forecast adequacy}, which is closely related to the statistical notion of conditional forecast calibration; see \citet{GR_2023} for a recent and detailed treatment of forecast calibration.
The strongest notion of \textit{ideal} forecasts $\widehat{\VaR}_{t,\beta}$ and $\widehat{\CoVaR}_{kt,\a|\b}$ for the VaR and CoVaR is given by
\begin{equation*}
	H_0^{\CoVaR} \colon \widehat{\VaR}_{t,\b} = \VaR_{t,\b}
	\qquad \text{and} \qquad 
	\widehat{\CoVaR}_{kt,\a|\b} = \CoVaR_{kt,\a\mid\b}, \quad k \in [K],\quad t \in \mathbb{N}.
\end{equation*}
We synonymously refer to this null as testing ideal forecasts, correct specification or calibration.

Clearly, $H_0^{\CoVaR}$ is not directly testable, because the true $\VaR_{t,\b}$ and $\CoVaR_{kt,\a|\b}$ are not observable even \textit{ex post}.
To circumvent this, we exploit the existence of a strict identification function for the CoVaR (jointly with the VaR), given in \citet[Theorem~S.3.1]{FH21},
\begin{align}
	\label{eqn:CoVaRJointID}
	\mV \big((v,c)', (x,y)' \big) = 
	\begin{pmatrix}\1_{\{x\leq v\}}-\beta\\ \1_{\{x> v\}}\big[\1_{\{y\leq c\}}-\alpha\big]\end{pmatrix}.
\end{align}
Identification functions are at the heart of many classical backtests \citep{NZ17}.
The property that renders them useful in testing calibration is that their (conditional) expectation is zero \textit{if and only if} the true (conditional) VaR and CoVaR are inserted.
More precisely, \citet[Theorem~S.3.1]{FH21} show that 
\begin{align*}
	\E \Big[\mV\big( ( v,c)' , (X_t, Y_{kt})' \big) \;\Big|\; \mathcal{F}_{t-1} \Big] = \vzeros \qquad\Longleftrightarrow\qquad 
	v = \VaR_{t,\b}  \;  \text{ and } \;   c = \CoVaR_{kt,\a|\b}
\end{align*}
under our conditions on the CDF of $(X_t, Y_{kt})^\prime\mid\mathcal{F}_{t-1}$.
Hence, the relevant implication in monitoring calibration
becomes
\begin{align}\label{eq:5p}
	 \E \left[\mV\left(\begin{pmatrix}\widehat{\VaR}_{t,\beta}\\\widehat{\CoVaR}_{kt,\a|\b}\end{pmatrix},\begin{pmatrix} X_t\\Y_{kt}\end{pmatrix}\right) \; \Bigg| \; \mathcal{F}_{t-1} \right]=\vzeros \qquad\text{for all}\ k \in [K], \; \, t\in \mathbb{N}.
\end{align}

The following result shows that the full, joint probabilistic properties of the indicators
\begin{align}
	\label{eqn:Indicators}
	I_{t} := \1_{\{X_t > \widehat{\VaR}_{t,\beta}\}}
	\qquad \text{and} \qquad
	I_{kt} := \1_{\{X_t> \widehat{\VaR}_{t,\beta},\ Y_{kt}> \widehat{\CoVaR}_{kt,\a|\b}\}}
\end{align}
is known under the null hypothesis $H_0^{\CoVaR}$.

\begin{prop}
	\label{prop:CoVaR null}
	Under $H_0^{\CoVaR}$, it holds for all $k \in [K]$ that
	\begin{align}
		\label{eqn:CoVaRH0Law}
		\big\{(I_t,I_{kt})^\prime\big\}_{t\in\mathbb{N}} \overset{d}{=} \big\{(\1_{\{U_{1t}>\beta\}},\1_{\{U_{1t}>\beta,\ U_{2t}>\alpha\}})^\prime\big\}_{t\in \mathbb{N}},
	\end{align}
	where $U_{it}\overset{IID}{\sim}\mathcal{U}[0,1]$ are independent of each other for  $i=1,2$ and all $t \in \mathbb{N}$.
\end{prop}

As the full probabilistic structure of $\big\{(I_t,I_{kt})^\prime\big\}_{t\in\mathbb{N}}$ is known for each $k \in [K]$ individually, the result of Proposition~\ref{prop:CoVaR null} can be used for our monitoring procedure to simulate critical values that are valid in a monitoring sense~\eqref{eqn:MonitoringGuarantee} without the need for asymptotic approximations.
Importantly, the binary nature of the CoVaR identification function in \eqref{eqn:CoVaRJointID} facilitates such a treatment, while remaining agnostic about the (conditional) distributions of $X_t$ and $Y_{kt}$.
Since for $K > 1$, Proposition~\ref{prop:CoVaR null} stays silent on the dependence structure between the contemporaneous $I_{kt}$ and $I_{k't}$ with $k \neq k'$, we will exploit a Bonferroni-type correction in Theorem~\ref{thm:CoVaR monitor} below.

For a (prospective) monitoring sample of size $n \in \mathbb{N}$, Proposition~\ref{prop:CoVaR null} leads to the testable implications of $H_0^{\CoVaR}$ that
\begin{equation}
	\label{eqn:Bernoulli}
	I_{t} \overset{\text{IID}}{\sim}\Bern(1-\b) \quad \text{and} \quad I_{kt} \overset{\text{IID}}{\sim}\Bern\big((1-\a)(1-\b)\big), \quad k \in [K],\ t \in [n],
\end{equation}
where $\Bern(\gamma)$ denotes a Bernoulli distribution with success probability $\gamma\in(0,1)$.

\begin{rem}[VaR monitoring for the CoVaR]
	\label{rem:VaRMonitoring}
	In CoVaR monitoring, it is crucial \emph{not} to disregard the VaR indicator $I_t$ in \eqref{eqn:Bernoulli}, even though superficially it only concerns the VaR while not being directly related to the CoVaR.
	To see why, consider probability levels $\a^\prime\neq\a$ and $\b^\prime\neq\b$ satisfying $(1-\a^\prime)(1-\b^\prime)=(1-\a)(1-\b)$. 
	Then, by the same arguments leading to \eqref{eqn:Bernoulli}, $I^\prime_{kt}:=\1_{\{X_t> \VaR_{t,\beta^\prime},\ Y_{kt}> \CoVaR_{kt,\a^\prime|\b^\prime}\}}\overset{\text{IID}}{\sim}\Bern\big((1-\a^\prime)(1-\b^\prime)\big) = \Bern\big((1-\a)(1-\b)\big)$. 
	Therefore, the true $\CoVaR_{kt,\a|\b}$ are indistinguishable from the incorrect $\CoVaR_{kt,\a^\prime|\b^\prime}$, because $I_{kt}$ and $I_{kt}^\prime$ are both IID with the same $\Bern\big((1-\a)(1-\b)\big)$-distribution. This would allow the possibility of completely misspecified CoVaR forecasts not being detected by a monitoring procedure that disregards $I_t$.
	Thus, it is vital to also monitor for VaR calibration via $I_{t} \overset{\text{IID}}{\sim}\Bern(1-\b)\neq\Bern(1-\b^\prime)$.
\end{rem}

We now explain the construction of a VaR detector that monitors $I_{t} \overset{\text{IID}}{\sim}\Bern(1-\b)$ from \eqref{eqn:Bernoulli}, and continue with a CoVaR detector below.
For monitoring the VaR-specific sequence $I_t$, we adapt the method of \citet{HD22a+}
and use a moving sum (MOSUM) detector inspired by the one-shot backtest of \citet{KW15} of the form
\begin{equation}
	\label{eq:VaR det}
	\VaR(T)=a\VaR^{uc}(T) + (1-a)\VaR^{iid}(T),
\end{equation}
where $T \in \{m,\dots,n\}$ and $a\in(0,1)$ with canonical choice $a = 0.5$. 
The constant $a$ is a user-specified parameter that allows to vary the sensitivity of the procedure towards violations of correct unconditional calibration (i.e., $\E[I_t]=1-\b$) and violations of IIDness of the $I_t$. These two properties are assessed by $\VaR^{uc}(T)$ and $\VaR^{iid}(T)$, respectively. 
The length of the rolling monitoring window is denoted by $m\in\mathbb{N}$, but the dependence of $\VaR(T)$ on $m$ is not made explicit for notational brevity.

For $\VaR^{uc}(T)$ in \eqref{eq:VaR det}, we choose the standardized detector
\begin{align}
	\label{eqn:VaRDetStandardization}
	\VaR^{uc}(T) = \frac{V_T- \E_{H_0^{\CoVaR}}[V_T] }{ \sqrt{\Var_{H_0^{\CoVaR}}(V_T) }}
	\quad \text{with} \quad 
	V_T = \bigg|\frac{1}{m}\sum_{t=T-m+1}^{T}I_{t}-(1-\b)\bigg|,\qquad T\geq m.
\end{align}
The mean, $\E_{H_0^{\CoVaR}}[V_T]$, and variance, $\Var_{H_0^{\CoVaR}}(V_T)$, are obtained from simulations as the full probabilistic structure of $I_t$ is known under $H_0^{\CoVaR}$; see \eqref{eqn:Bernoulli}.
The standardization in $\VaR^{uc}(T)$ differs from \citet{HD22a+} and is used to better balance the sensitivity of the two detectors in \eqref{eq:VaR det}.
In contrast to $\E_{H_0^{\CoVaR}}[V_T]$ and $\Var_{H_0^{\CoVaR}}(V_T)$, the sequence $V_T$ has to be computed from the data (i.e., observations and appertaining forecasts) and is responsible for giving the procedure its power.
Specifically, $V_T$ is designed to uncover deviations from $\E[I_t]=1-\b$, i.e., correct unconditional calibration.
By construction, large values of $\VaR^{uc}(T)$ provide evidence against the null, where deviations in both directions (i.e., $\E[I_t]>1-\b$ and $\E[I_t]<1-\b$) raise an alarm.\footnote{
	One could also consider one-sided detectors here. However, we refrain from doing so, because it remains unclear how a too conservative (and, hence, classically acceptable) misspecification of VaR forecasts interact with the validity of the CoVaR forecasts in light of the joint identification function in \eqref{eqn:CoVaRJointID}; also see \citet[Remark~2]{WWZ23}.
}

For $\VaR^{iid}(T)$ in \eqref{eq:VaR det}, we consider the durations $d_i:=t_i-t_{i-1}$ ($t_0:=T-m$) between VaR violation times in the window $\{T-m+1,\ldots,T\}$.
These are defined as
\begin{align*}
	\{t_1,\ldots,t_{S_{T}}\} =\big\{t\in\{T-m+1,\ldots,T\}\ :\ I_{t}=1\big\},
\end{align*}
where $S_{T} := \sum_{t=T-m+1}^{T}I_t$ counts the number of VaR violations in the window $\{T-m+1,\ldots,T\}$.
As a measure of the ``inequality'' between the durations $d_i$, we follow \citet{KW15} and use the Gini coefficient 
\begin{align*}
	g_{T} &= \frac{\frac{1}{S_{T}^2}\sum_{i,j=1}^{S_{T}}|d_i-d_j|}{2\overline{d}},
	\qquad \text{with} \qquad
	 \overline{d}=\frac{1}{S_{T}}\sum_{i=1}^{S_{T}}d_i.
\end{align*}
The detector for the second part of \eqref{eq:VaR det} then is the standardized Gini coefficient
\begin{align*}
	\VaR^{iid}(T) = \frac{g_{T} -  \E_{H_0^{\CoVaR}}[g_T]}{ \sqrt{\Var_{H_0^{\CoVaR}}(g_T) }},
\end{align*}
where $\E_{H_0^{\CoVaR}}[g_T]$ and $\Var_{H_0^{\CoVaR}}(g_T)$ are the mean and variance, respectively, of $g_T$ under the null.
As above, these moments are derived from Monte Carlo simulations, and it is only $g_T$ itself that is computed from the data.

The idea underlying the use of $\VaR^{iid}(T)$ is to test the IID~property of $I_t$ by considering the spacing between the VaR exceedances (i.e., those $t$ for which $I_t=1$). Unevenly spaced exceedances with large values for the Gini coefficient (and, hence, large values of $\VaR^{iid}(T)$) weigh against IIDness. 
Note that in principle also too evenly spaced exceedances (with small values for the Gini coefficient) provide evidence against $H_0^{\CoVaR}$. 
However, such a violation of the null is of no concern in risk management, where only a clustering of exceedances is seen as problematic \citep[see, e.g.,][Fig.~1]{KW15}.

To also monitor the CoVaR forecasts, i.e., to sequentially test $I_{kt}\overset{\text{IID}}{\sim}\Bern\big((1-\a)(1-\b)\big)$ from \eqref{eqn:Bernoulli}, we use the same strategy as for monitoring VaR predictions. 
Therefore, we define the CoVaR detectors $\CoVaR_{k}(T)$ for any $k \in [K]$ similarly as the VaR detector $\VaR(T)$:
\begin{align}
	\label{eq:CoVaR det}
	\CoVaR_{k}(T)=a\CoVaR^{uc}_{k}(T) + (1-a)\CoVaR^{iid}_{k}(T).
\end{align}
Here, the right-hand side quantities are defined in analogy to those in \eqref{eq:VaR det}, with the only difference that $I_{kt}$ replaces $I_t$, and $(1-\a)(1-\b)$ replaces $1-\b$ at every occurrence. 

The following theorem shows how the probability of a false detection in the sense of \eqref{eqn:MonitoringGuarantee} can be bounded uniformly over time and over different institutions $k \in [K]$, for which the CoVaR is forecasted.

\begin{thm}
	\label{thm:CoVaR monitor}
	For any significance level $\iota \in (0,1)$, it holds that
	\begin{multline}\label{eqn:CoVaRMonitoringTypeIError}
		\P_{H_0^{\CoVaR}} \Big\{ \exists T \in \{m, \dots, n\}: \quad \VaR(T) \ge v  \\ \text{  or  }  \;   \CoVaR_{k}(T) \ge c_k \quad \text{for some $k \in [K]$}  \Big\} \le \iota,
	\end{multline}
	if the critical values $v$ and the $c_k$'s are chosen such, that
	\begin{multline}
		\P_{H_0^{\CoVaR}}\Big\{\sup_{T=m,\ldots,n}\VaR(T) \geq v \Big\} + \sum_{k=1}^{K}\P_{H_0^{\CoVaR}}\Big\{\sup_{T=m,\ldots,n}\CoVaR_{k}(T) \geq c_k\Big\} \\
		- \sum_{k=1}^{K}\P_{H_0^{\CoVaR}}\Big\{\sup_{T=m,\ldots,n}\VaR(T) \geq v,\ \sup_{T=m,\ldots,n}\CoVaR_{k}(T) \geq c_k\Big\}=\iota.
		\label{eq:size ineq}
	\end{multline}
\end{thm}

In other words, Theorem~\ref{thm:CoVaR monitor} shows that size at level $\iota\in(0,1)$ is controlled \textit{in finite samples} if we reject $H_0^{\CoVaR}$ as soon as either
	\begin{equation*}
		\VaR(T) \geq v \quad\text{or}\quad \CoVaR_{k}(T) \geq c_k\qquad\text{for some}\ T=m,\ldots,n,\quad k \in [K].
	\end{equation*}

The two probabilities in the upper row of \eqref{eq:size ineq} are associated with a standard Bonferroni correction for testing the $K+1$ hypotheses of calibration of $\widehat{\VaR}_{t,\b}$ and $\widehat{\CoVaR}_{kt,\a|\b}$ for $k \in [K]$. 
However, the presence of the \textit{negative} third term on the right-hand side of \eqref{eq:size ineq} allows for a less conservative---and, hence, potentially more powerful---procedure.
We mention that there may not exist $v$ and $c_k$, such that the equality in \eqref{eq:size ineq} holds exactly, because of the binary nature of the detectors.
Therefore, in practice $v$ and $c_k$ should be chosen to lead to a sum that is smaller but as close as possible to $\iota$; see also step~\ref{it:final} in Algorithm~\ref{algo:1} below.

For the practical computation of the critical values $v$ and $c_k$, we use the following Algorithm~\ref{algo:1}, which approximates the probabilities in \eqref{eq:size ineq} by using the probabilistic structure in \eqref{eqn:CoVaRH0Law} to sample under the null.

\begin{algorithm}
	\label{algo:1}
	To compute critical values for Theorem~\ref{thm:CoVaR monitor} proceed as follows:
	\begin{enumerate}
		\item
		\label{it:1} 
		Generate a large number $B$ of mutually independent samples $\big\{U_{1t}\overset{\text{IID}}{\sim}\mathcal{U}[0,1]\big\}_{t \in [n]}$ and $\big\{U_{2t}\overset{\text{IID}}{\sim}\mathcal{U}[0,1]\big\}_{t \in [n]}$.
		
		\item
		\label{it:2} 
		\sloppy
		Compute the $B$ sequences $\big\{I_{t}^\ast=\1_{\{U_{1t} > \b\}}\big\}_{t \in [n]}$ and $\big\{I_{1t}^{\ast}=\1_{\{U_{1t} > \b,\ U_{2t}>\a \}}\big\}_{t \in [n]}$.
		
		\item\label{it:3} 
		For $b=1,\ldots,B$, calculate: 
		\begin{itemize}
			\item[i.] 
			$\sup_{T=m,\ldots,n}\VaR^{b}(T)$, where $\VaR^{b}(T)$ is defined as $\VaR(T)$, except that $\{I_t\}$ is replaced by the $b$-th sample from $\{I_t^{\ast}\}$ from step \ref{it:2} of this algorithm;
			
			\item[ii.] 
			$\sup_{T=m,\ldots,n}\CoVaR^{b}(T)$, where $\CoVaR^{b}(T)$ is defined as $\CoVaR_{1}(T)$, except that $\{I_{1t}\}$ is replaced by the $b$-th sample from $\{I_{1t}^{\ast}\}$ from step \ref{it:2} of this algorithm.
		\end{itemize}
		
		\item On a fine grid for $\nu\in[0,1]$, compute the empirical $(1-\nu)$-quantiles of 
		\begin{itemize}
			\item[i.] $\big\{\sup_{T=m,\ldots,n}\VaR^{b}(T)\big\}_{b \in [B]}$, which we denote by $v(\nu)$, and 
			\item[ii.] $\big\{\sup_{T=m,\ldots,n}\CoVaR^{b}(T)\big\}_{b \in [B]}$, which we denote by $c(\nu)$.
		\end{itemize}
		
		\item\label{it:final} Find the value $\nu \in [0,1]$ for which
		\begin{multline*}
			\frac{1}{B}\sum_{b=1}^{B}\1_{\big\{\sup_{T=m,\ldots,n} \VaR^{b}(T)\geq v(\nu) \big\}} + \frac{K}{B}\sum_{b=1}^{B}\1_{\big\{\sup_{T=m,\ldots,n} \CoVaR^{b}(T)\geq  c(\nu)\big\}} \\
			- \frac{K}{B}\sum_{b=1}^{B} \1_{\big\{\sup_{T=m,\ldots,n}\VaR^{b}(T)\geq v(\nu),\ \sup_{T=m,\ldots,n}\CoVaR^{b}(T)\geq c(\nu)\big\}}
		\end{multline*}
		is equal to (or smaller than) $\iota$; cf.~\eqref{eq:size ineq}. The appertaining critical values will be denoted by
		$v^{\CoVaR}$ and $c^{\CoVaR}$.

	\end{enumerate}
\end{algorithm}

The above algorithm also overcomes the theoretical ambiguity that there exists a continuum of critical values $v$ and $c_k$ that ensure the probabilities in \eqref{eq:size ineq} sum to $\iota$. 
In principle, such a continuum allows to weight the importance of the VaR and CoVaR hypotheses.
However, the fourth step in Algorithm~\ref{algo:1} advocates a natural approach that balances the importance of VaR and CoVaR by choosing $v$ and $c_k$ to correspond to some $(1-\nu)$-quantile of the variables $\sup_{T=m,\ldots,n}\VaR(T)$ and $\sup_{T=m,\ldots,n}\CoVaR_{k}(T)$, respectively.
Then, $\nu$ in step~5 simply has to be chosen to ensure that the probabilities in \eqref{eq:size ineq} sum to $\iota$.
Note that this entails the computation of two critical values only---one for the VaR (written $v^{\CoVaR}$) and one for the CoVaR (written $c^{\CoVaR}$), where the latter one is the same for all $k \in [K]$, because the distribution of $I_{kt}\overset{\text{IID}}{\sim}\Bern\big((1-\a)(1-\b)\big)$ is independent of $k$.
These critical values have the superscript ``${\CoVaR}$'' to distinguish them from critical values used for monitoring different systemic risk measures.

Clearly, the use of Boole's inequality for $K > 1$ (in the proof of Theorem~\ref{thm:CoVaR monitor}) implies that our monitoring procedure is conservative, indicated by the \emph{inequality} in  \eqref{eqn:CoVaRMonitoringTypeIError}.
Only for $K=1$, this becomes an equality such that size is kept exactly.
While a conservative test reduces the probability of actually making a type I error, this usually has the drawback of increasing the probability of a type II error, such that the ability to identify (systemic) risk changes is reduced. 
However, we show in simulations in Section~\ref{Simulations} that our procedure is not too conservative for values of $K$ representative of practical applications.

A possible reason for the good performance of the Bonferroni correction is explained by \citet{Hol79}, who states that: ``The power gain obtained by using a sequentially rejective Bonferroni test [based on ordering $p$-values] instead of a classical Bonferroni test depends very much upon the alternative. 
It is small if all the hypotheses are `almost true', but it may be considerable if a number of hypotheses are `completely wrong'.'' 
In our case, systemic risk builds up slowly for all banks, rendering all hypotheses 'almost true' at first, thus leading to good detection properties of our simple Bonferroni-type corrections.

The use of the Bonferroni-type correction has two further advantages.
First, it ensures controlled size in finite samples.
In particular, we do not have to rely on large-sample asymptotics, which may be unreliable in backtesting contexts \citep{Hog19a+}.
Second, the Bonferroni correction allows us to attribute a rejection of the null to a specific institution.
For instance, if the detector $\CoVaR_{k^{\ast}}$ for some $k^\ast \in [K]$ is the first to raise an alarm, the rejection of the null can be pinpointed to bank $k^\ast$. 
Of course, precisely identifying the most vulnerable institution is of utmost importance in systemic risk analysis.\footnote{ 
	A further implication of $H_0^{\CoVaR}$ is that $I_s$ and $I_{kt}$ are independent for $s\neq t$. We do not explicitly reflect this implication in the construction of our detectors. While doing so may increase the power of the detector (depending, of course, on the specific type of alternative), this would again render infeasible the attribution of a rejection of $H_0^{\CoVaR}$ to a specific institution.
}

\begin{rem}[Flexibility of detector specification]
	Other CoVaR monitoring detectors $\VaR(T)$ and $\CoVaR_k(T)$ depending solely on $\big\{(I_t,I_{kt})^\prime\big\}_{t\in [n]}$ can be constructed that preserve the accumulated type I error control in \eqref{eqn:MonitoringGuarantee}; compare Algorithm~\ref{algo:1}.
	This flexibility enables the construction of detectors specifically tailored to the particular form of misspecification one aims to detect.
	The same observation applies to the monitoring procedures for the RCoVaR in Section~\ref{ReversedCoVaRMonitoring}, and CoES and MES in Appendix~\ref{sex:CoESandMES}.
\end{rem}

\subsection{CoVaR Monitoring with Reversed Conditioning}
\label{ReversedCoVaRMonitoring}

Section~\ref{CoVaR Monitoring} focuses on a monitoring procedure for the CoVaR in \eqref{eqn:DefCoVaR}, which measures the individual exposure of (e.g., financial institution) $Y_{kt}$ onto the joint (market) $X_t$.
Now, we consider monitoring of the RCoVaR in \eqref{eqn:DefRevCoVaR}, which measures the effect the individual (e.g., financial institution) $Y_{kt}$ has on the joint (market) $X_t$.

Similarly as above, for given VaR forecasts $\widehat{\VaR}_{kt,\beta}$ and RCoVaR forecasts $\widehat{\RCoVaR}_{kt,\a|\b}$, the null is
\begin{equation*}
	H_0^{\RCoVaR} \colon \widehat{\VaR}_{kt,\b} = \VaR_{kt,\b}\qquad	\text{and} \qquad
	\widehat{\RCoVaR}_{kt,\a|\b} = \RCoVaR_{kt,\a|\b}, \quad k \in [K], \quad t \in \mathbb{N}.
\end{equation*}

Employing analogous arguments as in Section~\ref{CoVaR Monitoring}, we rely on a testable implication of $H_0^{\RCoVaR}$ that is based on the indicators $I_{kt}^{v} := \1_{\{Y_{kt}> \widehat{\VaR}_{kt,\beta}\}}$ and $I_{kt}^{c} := \1_{\{Y_{kt}> \widehat{\VaR}_{kt,\beta},\ X_t> \widehat{\RCoVaR}_{kt,\a|\b}\}}$.
Specifically, we have the following analog to Proposition~\ref{prop:CoVaR null}.

\begin{prop}
	\label{prop:CoVaR rev null}
	Under $H_0^{\RCoVaR}$, it holds for all $k \in [K]$ that
	\[
		\big\{(I_{kt}^v,I_{kt}^c)^\prime\big\}_{t\in\mathbb{N}} \overset{d}{=} \big\{(\1_{\{U_{1t}>\beta\}},\1_{\{U_{1t}>\beta,\ U_{2t}>\alpha\}})^\prime\big\}_{t\in\mathbb{N}},
	\]
	where $U_{it}\overset{IID}{\sim}\mathcal{U}[0,1]$ are independent of each other for  $i=1,2$ and all $t \in \mathbb{N}$.
\end{prop}

The proof is similar to that of Proposition~\ref{prop:CoVaR null} and, hence, omitted.
Proposition~\ref{prop:CoVaR rev null} suggests the following surveillance approach.
VaR changes may be monitored via the detector $\VaR_{k}(T)$, which is defined in analogy to $\VaR(T)$ at \eqref{eq:VaR det} with $I_t$ replaced by $I_{kt}^{v}$.
Similarly, CoVaR surveillance now uses the detector $\RCoVaR_{k}(T)$, where the definition is identical to that of $\CoVaR_{k}(T)$ at \eqref{eq:CoVaR det} except that $I_{kt}^{c}$ is used in place of $I_{kt}$.
\color{black}

\begin{thm}
	\label{thm:CoVaR rev monitor}
	For any significance level $\iota \in (0,1)$, it holds that
	\begin{multline}
		\label{eqn:RevCoVaRMonitoringTypeIError}
		\P_{H_0^{\RCoVaR}} \Big\{ \exists T \in \{m,\ldots,n\}: \quad \VaR_{k}(T) \ge v_k  \\
		\text{  or  }   \RCoVaR_{k}(T) \ge c_k \quad \text{for some $k \in [K]$}  \Big\} \le \iota,
	\end{multline}
	if the critical values $v_k$ and $ c_k$ are chosen such, that
	\begin{equation}
		\label{eq:size ineq rev}
		\sum_{k=1}^{K}\P_{H_0^{\RCoVaR}}\bigg\{\Big\{\sup_{T=m,\ldots,n}\VaR_{k}(T) \geq v_k \Big\}\cup\Big\{\sup_{T=m,\ldots,n}\RCoVaR_{k}(T) \geq c_k \Big\}\bigg\}=\iota.
	\end{equation}
\end{thm}

As in Section~\ref{CoVaR Monitoring}, almost the full probabilistic structure of the detectors is known under $H_0^{\RCoVaR}$, such that the probabilities in \eqref{eq:size ineq rev} and, thereby, the critical values $v_k$ and $c_k$ can be computed explicitly via simulations.
For this, we employ Algorithm~\ref{algo:RevCoVaR} in Appendix~\ref{sec:Algorithms}, which is very similar to Algorithm~\ref{algo:1}, but takes into account the differences between \eqref{eq:size ineq} and \eqref{eq:size ineq rev}.
Again, Algorithm~\ref{algo:RevCoVaR} only computes two critical values---one for the VaR component ($v^{\RCoVaR}$), and one for the RCoVaR component ($c^{\RCoVaR}$).

As in Theorem~\ref{thm:CoVaR monitor}, our procedure is less conservative than a standard Bonferroni correction.
Recall that a Bonferroni correction suggests monitoring each of the $2K$ hypotheses at a significance level of $\iota/(2K)$, such that the total probability of rejection under the null is bounded from above by $\iota$.
In contrast, our monitoring scheme based on Theorem~\ref{thm:CoVaR rev monitor} is less conservative because the sum in \eqref{eq:size ineq rev} is less than the ``Bonferroni sum'', i.e.,
\begin{multline*}
	\sum_{k=1}^{K}\P_{H_0^{\RCoVaR}}\bigg\{\Big\{\sup_{T=m,\ldots,n}\VaR_{k}(T) \geq v_k \Big\}\cup\Big\{\sup_{T=m,\ldots,n}\RCoVaR_{k}(T) \geq c_k \Big\}\bigg\}\\
	\leq \sum_{k=1}^{K}\bigg[\P_{H_0^{\RCoVaR}}\Big\{\sup_{T=m,\ldots,n}\VaR_{k}(T) \geq v_k\Big\}+\P_{H_0^{\RCoVaR}}\Big\{\sup_{T=m,\ldots,n}\RCoVaR_{k}(T) \geq c_k \Big\}\bigg].
\end{multline*}

\begin{rem}[Model-estimation error]
	\label{rem:ModelEstimation}
	We intentionally exclude model-estimation error from the uncertainty quantification of our detectors. Consequently, we treat the forecasts as the outcome of a particular modeling choice---including its estimation---and therefore evaluate the forecasting model and the estimation method jointly as in \citet{DM95}, \citet{GW06} and \citet{RS19}. In this sense, the risk forecasts are regarded as the final product that must satisfy the respective null hypotheses, irrespective of the underlying model and its estimation. This approach parallels \citet{HD22a+}, who provide several arguments in favor of this practice for monitoring risk forecasts, all of which carry over directly to the present setting.
	See Figure~\ref{fig:PowerEstWindowCoVaR} for an analysis of our procedures' sensitivity to model-estimation error.
\end{rem}

\section{Simulations}
\label{Simulations}

Here, we investigate size and power of our monitoring procedures. 
We place a particular emphasis on studying size, because Bonferroni-type corrections---such as those underlying our monitoring schemes---are known to be conservative when many hypotheses are tested.

\subsection{The DCC--GARCH Data-Generating Process}
\label{Data-Generating Process}

We simulate the variables $\big\{(X_t,Y_{1t},\ldots,Y_{Kt})^\prime\big\}_{t=-E+1,\ldots,n}$ from a DCC--GARCH model of \citet{Eng02}.
The sequence consists of $E \in \mathbb{N}$ time points for model estimation (used in one simulation setting), and $n \in \mathbb{N}$ time points for generating forecasts and their evaluation.
The data-generating process (DGP) is
\begin{equation*}
	\mW_t:=(X_t,Y_{1t},\ldots,Y_{Kt})^\prime=\mD_t \vepsi_t,\qquad\vepsi_t\mid\mathcal{F}_{t-1}\sim t_{\nu}(\mR_t),
\end{equation*}
where $\mD_t^2$ is the $\mathcal{F}_{t-1}$-measurable diagonal matrix containing the componentwise conditional variances of $\mW_t\mid\mathcal{F}_{t-1}$, the matrix $\mR_t$ is the conditional correlation of $\mW_t\mid\mathcal{F}_{t-1}$, and $t_{\nu}(\mR_t)$ denotes the multivariate $t$-distribution with degrees of freedom equal to $\nu>2$, correlation matrix $\mR_t$ and marginals standardized to have zero mean and unit variance. 
Therefore, the conditional variance-covariance matrix of $\mW_t\mid\mathcal{F}_{t-1}$ is $\mH_t=\mD_t\mR_t\mD_t$. The matrices $\mD_t$ and $\mR_t$ are modeled in the DCC--GARCH fashion of \citet{Eng02} as
\begin{align*}
	\mD_t^2 &= \diag(\vomega_{G}) + \diag(\valpha_{G})\circ \mW_{t-1}\mW_{t-1}^{\prime} + \diag(\vbeta_{G})\circ\mD_{t-1}^2,\\
	\mQ_t &= \overline{\mQ}(1-\alpha_Q-\beta_Q) + \alpha_Q \vepsi_{t-1}\vepsi_{t-1}^\prime + \beta_Q \mQ_{t-1},\\
	\mR_t &= \diag(\mQ_t)^{-1/2}\mQ_t\diag(\mQ_t)^{-1/2},
\end{align*}
where $\overline{\mQ}$ is the unconditional correlation matrix of the ``devolatilized'' series $\vepsi_{t-1}$, and  $\diag(\va)$ and  $\diag(\mA)$ denote the diagonal matrices containing the elements of the vector $\va$ and the diagonal elements of the matrix $\mA$, respectively.
Note that both the conditional variance matrix $\mD_t^2$ and the conditional correlation matrix $\mR_t$ evolve in a GARCH-type fashion (the latter through the recursion for $\mQ_t$).

We parametrize the process as follows:
We choose $\overline{\mQ}$ to be the matrix with ones on the main diagonal and 0.5 elsewhere (i.e., an equicorrelation matrix), such that the unconditional correlation between any two elements of $\vepsi_t$ equals 0.5.
Furthermore, we fix the parameters $\nu=5$, $\vomega_G=(0.1,\ldots, 0.1)^\prime$, $\valpha_G=(0.1,\ldots,0.1)^\prime$, and $\alpha_{Q}=0.1$.
		In contrast, the autoregressive parameters $\vbeta_G=(\beta_{G,1},\ldots,\beta_{G,K+1})^\prime$ and $\beta_{Q}$ vary over time. 
		Specifically,
		\begin{align}
			\label{eqn:SimParamMisspec}
			\beta_{G,i} = \beta_{G,i,t} =\begin{cases}  0.7, & t\leq t^{\ast},\\
				\beta_\text{post}, & t> t^{\ast},	
			\end{cases}
			\qquad\text{and}\qquad \beta_Q = \beta_{Q,t} =\begin{cases}  0.7, & t\leq t^{\ast},\\
				\beta_\text{post}, & t> t^{\ast},	
			\end{cases}
		\end{align}
		such that there is an upward change in the persistences at time $t^{\ast} \in [n]$. 
		In the baseline simulation setup, we fix $\beta_\text{post} = 0.85$.
		
		Through \eqref{eqn:SimParamMisspec}, we misspecify the dynamics of the process, which will then be picked up by the monitoring procedure.
		Of course, if $t^{\ast}=n$, then there is no structural change in the sample, and the systemic risk forecasts (issued from the estimates of the initial fixed window) are correctly specified. 
		Hence, except for the estimation error, the forecasts are correctly specified, such that a rejection probability close to the nominal level should be expected.

		\subsection{Systemic Risk Forecasts from DCC--GARCH Models}
		\label{sec:SystemicRiskForecasts}
		
		For generating the forecasts, we either use fixed DCC--GARCH  parameters or (in one setup) estimate the parameters based on a sample $\big\{(X_t,Y_{1t},\ldots,Y_{Kt})^\prime\big\}_{t=-E+1,\ldots,0}$ of length $E$ by using  the \texttt{R} package \texttt{rmgarch} \citep{rmgarch}. 
		In doing so, we estimate the marginals via standard Gaussian quasi-maximum likelihood estimation (to obtain $K+1$ estimates of $\omega_{G,i}$, $\alpha_{G,i}$ and $\beta_{G,i}$). 
		The dependence parameters ($\nu$, $\alpha_Q$ and $\beta_Q$) are estimated in a second step using a multivariate $t$-assumption for the $\vepsi_t$.

		Given DCC-forecasts (either based on fixed or estimated model parameters) for $\widehat{\mD}_t$, $\widehat{\mR}_t$ and $\widehat{\mH}_t$, we obtain the VaR forecasts at time $t$ as $\widehat{\VaR}_{t,\b} = \widehat{\mD}_{t,11}^{-1} q_{t_\nu}(\beta)$, where $q_{t_\nu}(\beta)$ denotes the $\beta$-quantile of the univariate $t_{{\nu}}$-distribution with unit variance and (either fixed or estimated) degrees of freedom $\nu$.
		For the CoVaR forecasts, we are not aware of a closed-form solution.
		Instead, we apply a root finding algorithm to approximate $\widehat{\CoVaR}_{kt,\a|\b}$ via
		\begin{align}
			\label{eqn:CoVaRFCApprox} 
			\mathbb{P} 
			\Big\{ Y_{kt} \ge \widehat{\CoVaR}_{kt,\a|\b}, \; X_t \ge \widehat{\VaR}_{t,\b} \;\Big| \; \mathcal{F}_{t-1} \Big\} -  (1-\alpha) (1-\beta) = 0.
		\end{align}
		The probability in  \eqref{eqn:CoVaRFCApprox} is taken with  respect to $(X_t, Y_{kt})\mid\mathcal{F}_{t-1} \sim t_\nu(\widehat{\mH}_{t,k})$, where $\widehat{\mH}_{t,k} = \big( ( \widehat{\mH}_{t, 11}, \widehat{\mH}_{t, 1 (k+1)})'  \mid ( \widehat{\mH}_{t, 1 (k+1)}, \widehat{\mH}_{t, (k+1) (k+1)})' \big)$
		consists of the respective  entries of the forecasted variance-covariance matrix $\widehat{\mH}_{t}$.
		Based on the forecasts $\widehat{\VaR}_{t,\b}$ and $\widehat{\CoVaR}_{kt,\a|\b}$, we compute the VaR and CoVaR indicators as described in Sections~\ref{CoVaR Monitoring}--\ref{ReversedCoVaRMonitoring}.

		\subsection{Simulation Results under the Null Hypothesis}
		\label{sec:Simulation_Results_Null}

		All following results are based on 5000 simulation replications.
		We first fix $n=1000$, $t^\ast=n$, $m=250$, and generate the systemic risk forecasts based on the correct DCC--GARCH parameters from the ``pre-break'' period such that we omit parameter estimation noise.
		We deliberately do not compare against one-shot systemic risk backtests of, e.g., \citet{Bea21} and \citet{FH21} since these---opposed to our monitoring procedure---accumulate type I errors when applied repeatedly.

		\begin{table}[tb]
			\centering
			\resizebox{\linewidth}{!}{
			\small
			\begin{tabular}{llllrlrrrrrrrrrrrr}
				\toprule
				$\alpha = \beta$ &  & $K$ &  & Joint &  & VaR & & \multicolumn{10}{c}{CoVaR for series number $k$} \\
				\cmidrule{9-18}  
				&&&&&&&& 1 & 2 & 3 & 4 & 5 & 6 & 7 & 8 & 9 & 10 \\ 
				\midrule
				\multirow{4}{*}{0.9} &  & 1 &  & 10.00 &  & 5.30 && 4.70 &  &  &  &  &  &  &  &  &  \\ 
				&  & 2 &  & 10.28 &  & 3.30 && 3.46 & 3.56 &  &  &  &  &  &  &  &  \\
				&  & 5 &  & 7.58 &  & 1.46 && 1.14 & 1.32 & 1.40 & 1.10 & 1.32 &  &  &  &  &  \\ 
				&  & 10 &  & 8.36 &  & 0.78 && 0.92 & 0.74 & 0.96 & 0.80 & 0.64 & 0.54 & 0.88 & 0.86 & 0.78 & 0.58 \\ 
				\midrule
				\multirow{4}{*}{0.95} &  & 1 &  & 9.62 &  & 4.28 && 5.36 &  &  &  &  &  &  &  &  &  \\ 
				&  & 2 &  & 9.26 &  & 3.46 && 3.04 & 2.82 &  &  &  &  &  &  &  &  \\
				&  & 5 &  & 8.24 &  & 1.56 && 1.22 & 1.24 & 1.38 & 1.50 & 1.52 &  &  &  &  &  \\ 
				&  & 10 &  & 6.38 &  & 0.68 && 0.66 & 0.60 & 0.52 & 0.74 & 0.60 & 0.52 & 0.46 & 0.68 & 0.70 & 0.50 \\ 
				\bottomrule
			\end{tabular}
		}
			\caption{Joint and disaggregated \textit{first} rejection rates (in percent) of our CoVaR surveillance procedure under the null hypothesis of correctly specified forecasts. 
			Results are reported for a nominal level of $\iota = 10\%$, for $K \in \{1,2,5,10\}$ and $\alpha = \beta \in \{0.9, 0.95\}$.
			The column ``VaR'' reports the percentage of cases in which the VaR detector raised the first alarm, while the columns labeled ``1''--``10'' report the corresponding first-alarm rates of the respective CoVaR detectors.}
			\label{tab:SizeCoVaR}
		\end{table}

		\begin{table}[tb]
			\centering
			\resizebox{\linewidth}{!}{
			\small
			\begin{tabular}{llllrlllrrrrrrrrrr}
				\toprule
				&&&&&&&& \multicolumn{10}{c}{Series number $k$} \\
				\cmidrule{9-18}  
				$\alpha = \beta$ &  & $K$ &  & Joint &  & Measure &  & 1 & 2 & 3 & 4 & 5 & 6 & 7 & 8 & 9 & 10 \\ 
				\midrule
				&  & \multirow{2}{*}{2} &  & \multirow{2}{*}{9.32} &  & VaR &  & 2.22 & 2.48 &  &  &  &  &  &  &  &  \\ 
				&  & &  &  &  & RCoVaR &  & 2.32 & 2.44 &  &  &  &  &  &  &  &  \\[0.3em] 
				\multirow{2}{*}{0.9} &  & \multirow{2}{*}{5} &  & \multirow{2}{*}{6.96} &  & VaR &  & 0.70 & 0.88 & 0.62 & 0.92 & 1.06 &  &  &  &  &  \\ 
				&  &  &  &  &  & RCoVaR &  & 0.66 & 0.56 & 0.56 & 0.68 & 0.62 &  &  &  &  &  \\[0.3em]  
				&  & \multirow{2}{*}{10} &  & \multirow{2}{*}{7.70} &  & VaR &  & 0.48 & 0.40 & 0.50 & 0.44 & 0.36 & 0.48 & 0.64 & 0.62 & 0.54 & 0.50 \\ 
				&  &  &  & &  & RCoVaR &  & 0.30 & 0.34 & 0.28 & 0.52 & 0.34 & 0.30 & 0.30 & 0.30 & 0.36 & 0.26 \\ 
				\midrule
				&  & \multirow{2}{*}{2} &  & \multirow{2}{*}{9.64} &  & VaR &  & 2.80 & 2.92 &  &  &  &  &  &  &  &  \\ 
				&  & &  & &  & RCoVaR &  & 2.12 & 2.22 &  &  &  &  &  &  &  &  \\[0.3em]  
				\multirow{2}{*}{0.95} &  &  \multirow{2}{*}{5} &  & \multirow{2}{*}{8.12} &  & VaR &  & 1.00 & 1.02 & 1.20 & 0.88 & 1.18 &  &  &  &  &  \\ 
				&  &  &  &  &  & RCoVaR &  & 0.80 & 0.70 & 0.78 & 0.74 & 0.82 &  &  &  &  &  \\[0.3em]  
				&  &  \multirow{2}{*}{10} &  & \multirow{2}{*}{6.70} &  & VaR &  & 0.66 & 0.56 & 0.44 & 0.44 & 0.40 & 0.50 & 0.48 & 0.52 & 0.54 & 0.50 \\ 
				&  & &  & &  & RCoVaR &  & 0.26 & 0.30 & 0.22 & 0.22 & 0.20 & 0.22 & 0.28 & 0.22 & 0.42 & 0.12 \\ 
				\bottomrule
			\end{tabular}
		}
			\caption{Joint and disaggregated \textit{first} rejection rates (in percent) of our RCoVaR surveillance procedure under the null hypothesis of correctly specified forecasts. 
			Results are reported for a nominal level of $\iota = 10\%$, for $K \in \{2,5,10\}$ and $\alpha = \beta \in \{0.9, 0.95\}$.
			The columns labeled ``1''--``10'' report the percentage of cases in which the respective VaR or CoVaR detector---indicated in the column ``Measure''---raised the first alarm.}
			\label{tab:SizeRCoVaR}
		\end{table}

		The columns ``Joint'' in Tables~\ref{tab:SizeCoVaR} and \ref{tab:SizeRCoVaR} show the rejection rates of the CoVaR and RCoVaR monitoring procedures under the no-break null hypotheses---i.e., $t^\ast=n$ in \eqref{eqn:SimParamMisspec}.
		We consider $\alpha = \beta \in \{0.9, 0.95\}$ together with $K \in \{1,2,5,10\}$ for the CoVaR and  $K \in \{2,5,10\}$.
		For $K=1$, the rejection rates of the CoVaR and RCoVaR coincide due to the symmetry of the DGP.
		We find that all combined empirical rejection rates in the columns ``Joint'' are below $10.28\%$, which corresponds to our theoretical finding that the monitoring procedures hold size exactly.\footnote{Given that these empirical rejection rates are based on 5000 Monte Carlo replications, a $99\%$-confidence interval of IID Bernoulli-distributed variables with success probability 0.1 is given by approximately $[8.9\%, 11.1\%]$, indicating that our procedure is able to hold size \emph{exactly}, while the deviations from $10\%$ are explained by the finite number of Monte Carlo replications.}
		Even for $K=10$, the rejection rates are not too conservative and lie between $6\%$ and $9\%$, which illustrates that the application of Boole's inequality in the proof of Theorem~\ref{thm:CoVaR monitor} is relatively tight here.

		As the columns ``Joint'' merely indicate whether \emph{any} of the detectors raises an incorrect alarm under the null, we also analyze \emph{which} of the detectors is responsible for the rejection.
		For this, the right-hand side columns (denoted ``VaR'' and with the series numbers ``1''--``10'') of the tables show the frequencies how often a given detector is responsible for \emph{first} raising a false alarm.
		Notice here that the joint rejection frequency is somewhat lower than the sum of the individual frequencies as multiple detectors occasionally reject simultaneously.
		Overall, we find that for both, CoVaR and RCoVaR, the spurious alarms of the individual detectors are relatively balanced between the VaR and the systemic risk detectors as well as between the different assets $k=1,\dots,K$.
		Especially the former is noteworthy, and as desired by the choice of the level $\nu$ in steps 4 and 5 of the Algorithms~\ref{algo:1} and \ref{algo:RevCoVaR}.

\subsection{Simulation Results under the Alternative Hypothesis}
\label{sec:Simulation_Results_Alternative}		

We now analyze our procedures' power to detect misspecified forecasts. 
For this, Figure~\ref{fig:PowerCoVaR} plots rejection frequencies under the alternative ($t^\ast<n$ in \eqref{eqn:SimParamMisspec}).
The plots in the left panel display power against the break point $t^\ast \in \{0,1, \dots, 1000\}$, with $\beta_{\text{post}} = 0.85$ in \eqref{eqn:SimParamMisspec}.
The right panel shows power for a fixed break point at  $t^\ast = 0$ for a varying degree of post-break parameter misspecification $\beta_\text{post} \in \{0.7, 0.75, 0.8, 0.85, 0.87, 0.89, 0.899\}$.
We again consider $K\in\{1,2,5,10\}$ and $\alpha = \beta \in \{0.9, 0.95\}$.
Notice that monitoring only starts at $t=m=250$.
In the left panel, the plots recover the ``joint'' size from Tables~\ref{tab:SizeCoVaR}--\ref{tab:SizeRCoVaR} for $t^\ast = n = 1000$, and in the right panel, for $\beta_\text{post} = 0.7$.

\begin{figure}[tb]
	\centering
	\includegraphics[width=\textwidth]{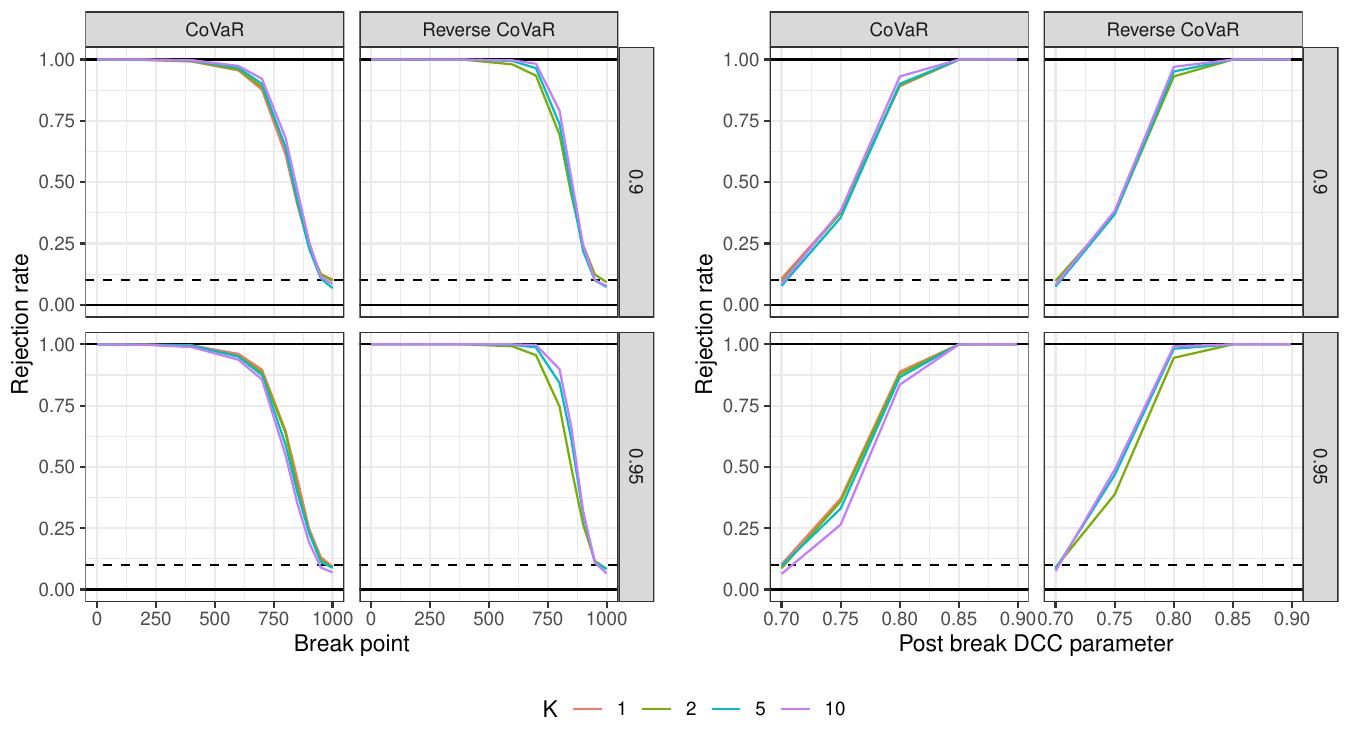}
	\caption{Rejection rates of our CoVaR and RCoVaR surveillance methods plotted against the break point in \eqref{eqn:SimParamMisspec} in the left panel and against the post-break DCC parameter in the right panel.
	In both plots, we consider $\alpha = \beta \in \{0.9, 0.95\}$ as well as $K \in \{1,2,5,10\}$ for the CoVaR and  $K \in \{2,5,10\}$ for the RCoVaR.}
	\label{fig:PowerCoVaR}
\end{figure}

As expected, power increases monotonically for earlier break points as well as for higher degrees of parameter misspecification in all plots, and the procedure works equally well for both probability levels.
Whether a larger number $K$ of to-be-monitored stocks increases or decreases power depends on the specific situation.
In general, monitoring more sequences simultaneously is subject to a trade-off between a stricter correction of the significance level and more series in which misspecifications can be detected.

When comparing the monitoring procedures \emph{across} systemic risk measures, we find that the reverse CoVaR procedure is the most powerful, which can be explained by the fact that more ($2K$ instead of $K+1$) forecast series are monitored, and that more ($K$ instead of 1) VaR series are monitored simultaneously, which are not ``as far in the tail'' as the CoVaR.

\begin{figure}[tb]
	\centering
	\includegraphics[width=0.75\textwidth]{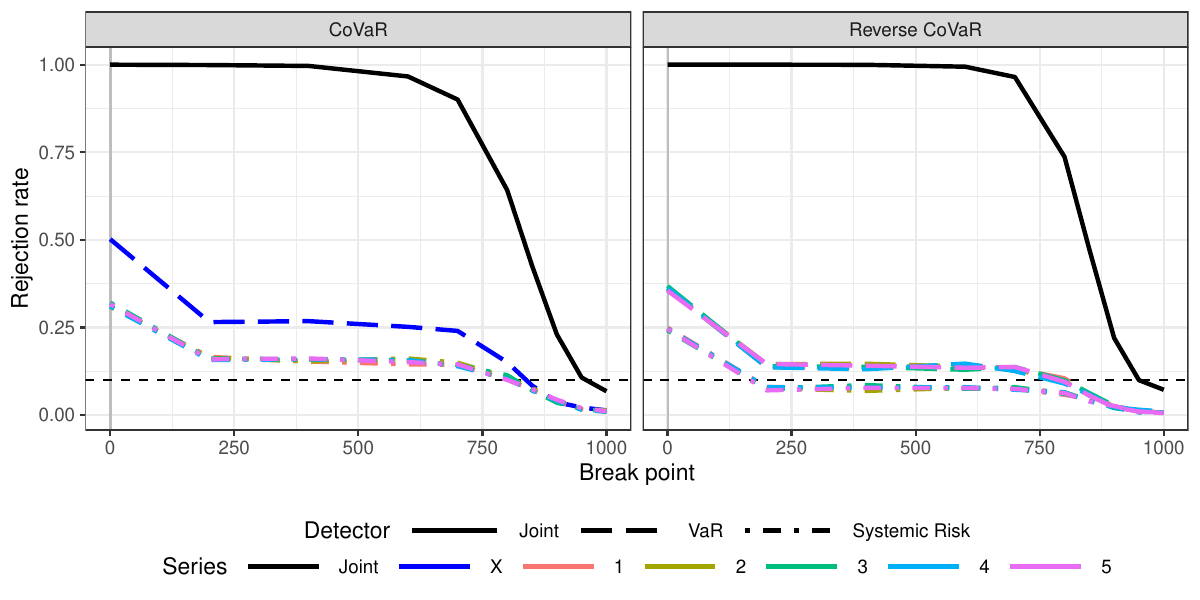}
	\caption{Rejection rates of the joint CoVaR and RCoVaR procedure in black, and frequencies how often the individual detectors are the \emph{first} to generate a detection in colors in the setting of the left panel of Figure~\ref{fig:PowerCoVaR}.
	Dashed lines depict ``first'' rejection rates for the VaR detector, and dot-dashed lines for the systemic risk detectors. The colors indicate the respective component of $\mW_t$. 
	The nominal level of $\iota = 10\%$ is indicated by the dashed horizontal line.}
	\label{fig:PowerCoVaRIndividual}
\end{figure}

Figure~\ref{fig:PowerCoVaRIndividual} shows the joint rejection rates in the same setting as in the left panel of Figure~\ref{fig:PowerCoVaR} in black, keeping $K=5$ and $\alpha = \beta = 0.9$ fixed.
Additionally, the colored dashed (VaR) and dot-dashed (Systemic Risk) lines indicate how often the respective detector was the first to raise an alarm, akin to the numbered columns in Tables~\ref{tab:SizeCoVaR}--\ref{tab:SizeRCoVaR}.
We find that the systemic risk detectors for different $k \in [K]$ have identical rates of rejecting first, which is sensible given the symmetry of the DCC--GARCH DGP.
Moreover, the VaR detector is more powerful than the systemic risk detectors, most likely as the VaR is not as far in the tail.

\begin{figure}[tb]
	\centering
	\includegraphics[width=\textwidth]{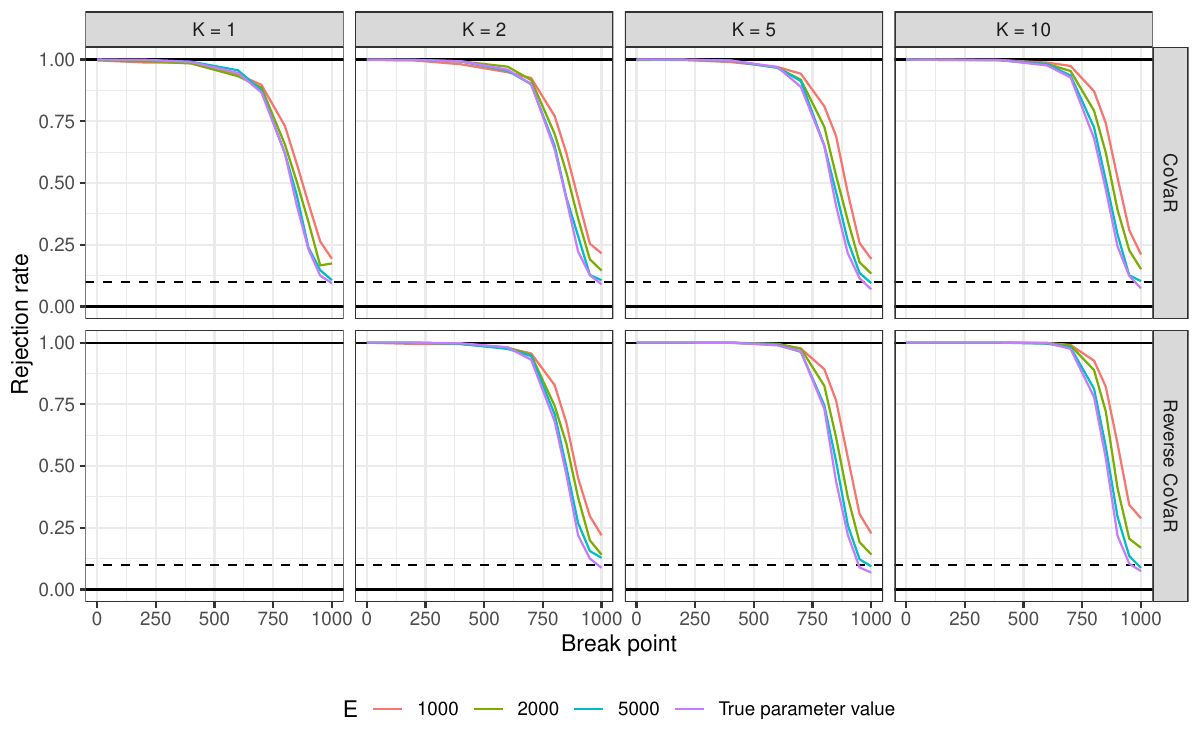}
	\caption{Rejection rates of our CoVaR and RCoVaR surveillance methods plotted against the break point for different estimation window lengths $E$, where $E = \infty$ implies the use of the correct parameters.
	We consider $K \in \{1,2,5,10\}$ and note that the rejection rates of CoVaR and RCoVaR coincide for $K=1$.
	We further keep $n=1000$, $t^\ast = 0$,  $\beta_\text{post} = 0.85$ and $\alpha = \beta = 0.9$ fixed.}
	\label{fig:PowerEstWindowCoVaR}
\end{figure}

We continue to analyze the effect parameter estimation noise within the forecasts has on the rejection frequencies. 
Recall that, as argued in Remark~\ref{rem:ModelEstimation}, we view model estimation error as part of a misspecified forecast sequence.
Figure~\ref{fig:PowerEstWindowCoVaR} compares the rejection rates when the forecasts are based on (correctly specified) DCC--GARCH models that are estimated on an in-sample period of length $E \in \{1000, 2000, 5000, \infty\}$ with $K \in \{1,2,5,10\}$, while fixing $n=1000$, $m=250$, $t^\ast = 0$, $\beta_\text{post} = 0.85$, and $\alpha = \beta = 0.9$.
We find that small(er) estimation window sizes distort the systemic risk forecasts and, hence, deliver increased rejection rates of up to $20\%$ even in the ``no-break case'' with $t^\ast = 1000$.
This effect is somewhat more pronounced for the RCoVaR than for the CoVaR.
A further effect of the model estimation is that power increases naturally for  $t^\ast < 1000$ when the length of the estimation period $E$ is decreased.

Overall, our simulations reinforce the theoretical finding that the surveillance schemes hold size exactly, even when the tests are employed repeatedly at every time point.
Furthermore, the power of our tests behaves naturally for varying break times, different break magnitudes, and a varying dimensionality of the monitored sequences.
Finally, we find that VaR and systemic risk rejections occur relatively balanced under the null, whereas the VaR detector naturally has more power (i.e., it detects earlier), due to the systemic risk measures being further out in the tail.

\section{Empirical Application}
\label{Empirical Application}
		
		We apply the systemic risk surveillance procedures to real financial data to analyze their sensitivity to sub-optimal risk forecasts in practice.
		For the market indicator $X_t$, we use the negative returns of the S\&P\,500 Financials index (SPF) and we use the $K=5$ systemically important US banks Bank of America Corp (BAC), Citigroup Inc (C), Goldman Sachs Group Inc (GS), JPMorgan Chase \& Co (JPM) and Wells Fargo \& Co (WFC) as $(Y_{1t}, \dots, Y_{Kt})^\prime$.
		The four considered monitoring time periods each span $n=1000$ days and are chosen as 
		(i) the global financial crisis: 23 September 2005 -- 14 September 2009, containing the bankruptcy of Lehman Brothers on 15 September 2008;
		(ii) a calm period: 14 January 2013 -- 30 December 2016;
		(iii) the COVID period: 23 March 2017 -- 12 March 2021, containing the  outbreak of the COVID pandemic; and
		(iv) Trump's tariffs: 29 October 2021 -- 23 October 2025, containing the introduction of Trump's tariffs on 2 April 2025 (and the collapse of the Silicon Valley Bank on 10 March 2023).
		In all four settings, we use a rolling monitoring window of $m=250$ trading days, consider the systemic risk measures at levels $\alpha = \beta = 0.95$ and use the monitoring level of $\iota = 0.1$.
		
		\begin{figure}[tb]
			\centering
			\includegraphics[width=\textwidth]{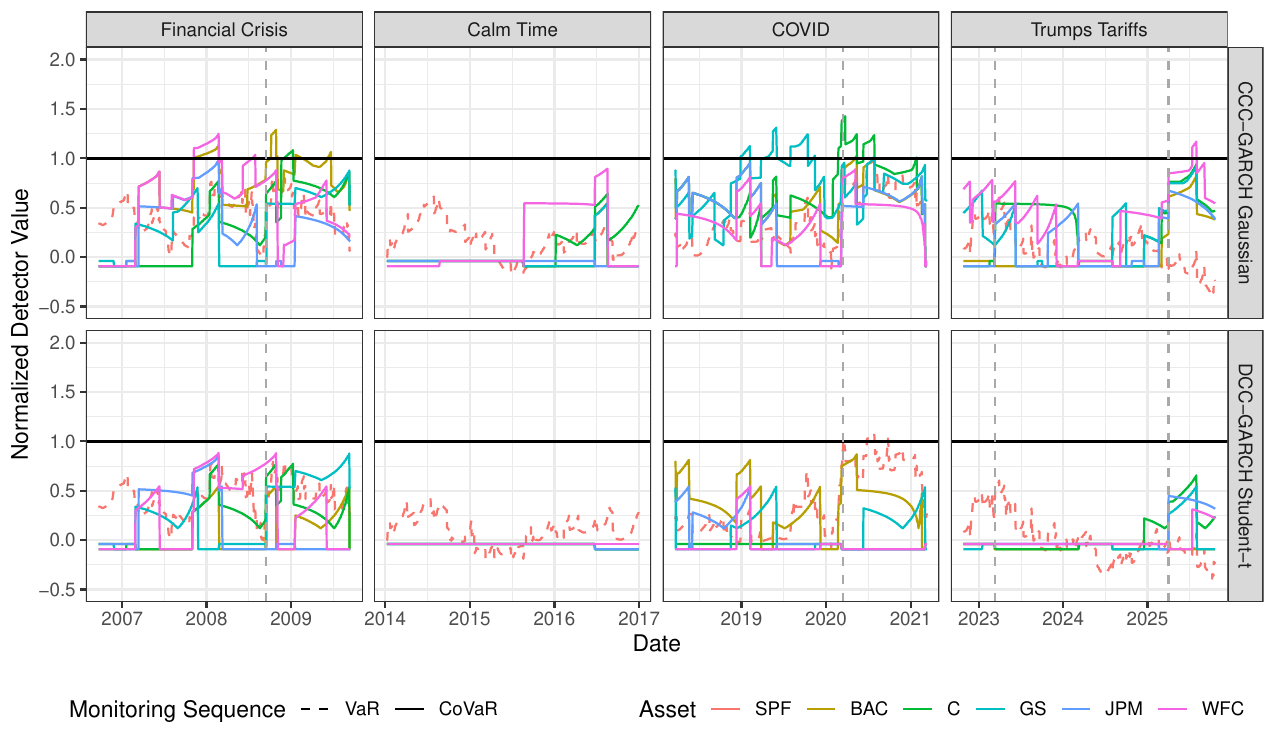}
			\caption{
				Normalized VaR and CoVaR detector values for the SPF, and the financial institutions BAC, C, GS, JPM and WFC for forecasts from a Gaussian CCC--GARCH and a Student's $t$ DCC--GARCH model.
				We use the specifications $\iota=0.1$, $\alpha=\beta=0.95$,  $E = 1500$, $n=1000$ and $m=250$. 
				The displayed detector values are normalized by their respective critical values, such that a  detector exceeding the black horizontal unit line implies a detection.
				The vertical dotted lines represent the bankruptcy of Lehman Brothers on 15 September 2008, 
				the beginning of the COVID crisis on 13 March 2020 (when the US declared a national emergency), 
				the collapse of the Silicon Valley Bank on 10 March 2023, and 
				the tariff announcement of Donald Trump on 2 April 2025.
			}
			\label{fig:Detectors_CoVaR}
		\end{figure}

		We generate systemic risk forecasts by modeling the $K+1=6$ returns by a CCC--GARCH model with Gaussian innovations and a DCC--GARCH model with Student's $t$ innovations, introduced by \citet{Eng02} and further described in Section~\ref{Data-Generating Process}.
		We deliberately choose a Gaussian distribution for the CCC--GARCH model's innovations to analyze a suboptimal model.
		The latter DCC--GARCH--$t$ model performs relatively well in horse races of multivariate volatility models \citep{LRV12,LRV13} and is, hence, still a standard forecasting model in financial risk management.
		From these two models, we generate CoVaR and RCoVaR forecasts as described in Section~\ref{sec:SystemicRiskForecasts}  based on the presumed distributions of the model innovations $\vepsi_t\mid\mathcal{F}_{t-1}\sim t_{\nu}(\mR_t)$; using $\nu = \infty$ for the Gaussian case.
		The models are estimated in all four settings using $E = 1500$ observations, i.e., data starting approximately six years before the beginning of the monitoring period.
		
		Figure~\ref{fig:Detectors_CoVaR} presents the VaR and CoVaR detector values, normalized by their respective critical values, for both multivariate GARCH models across the four considered time periods.
		Under this normalization, values exceeding unity indicate rejections of the null hypothesis.
		The detector values may take negative values as a consequence of their standardization, as defined, for example, in \eqref{eqn:VaRDetStandardization}.
		
		While we do not find rejections during the calm time for either forecasting model, the detectors raise an alarm for the CCC--GARCH--$\mathcal{N}$ forecasts in all three other considered periods.
		Importantly, these alarms are raised by the systemic CoVaR detector, hence implying that specifically the systemic component of the risk forecasts is misspecified.
		This highlights the importance of monitoring \textit{systemic} risk as opposed to monitoring only the risk component (i.e., the VaR) via the procedures of \citet{HD22a+} and \citet{WWZ23}.
		In contrast, for the DCC--GARCH--$t$ forecasts, only the VaR component raises an alarm after the start of the COVID pandemic, showing that the dynamic correlation structure implies by the DCC is much better able to capture the changing co-movements in turbulent times.

	\begin{figure}[tb]
			\centering
			\includegraphics[width=\textwidth]{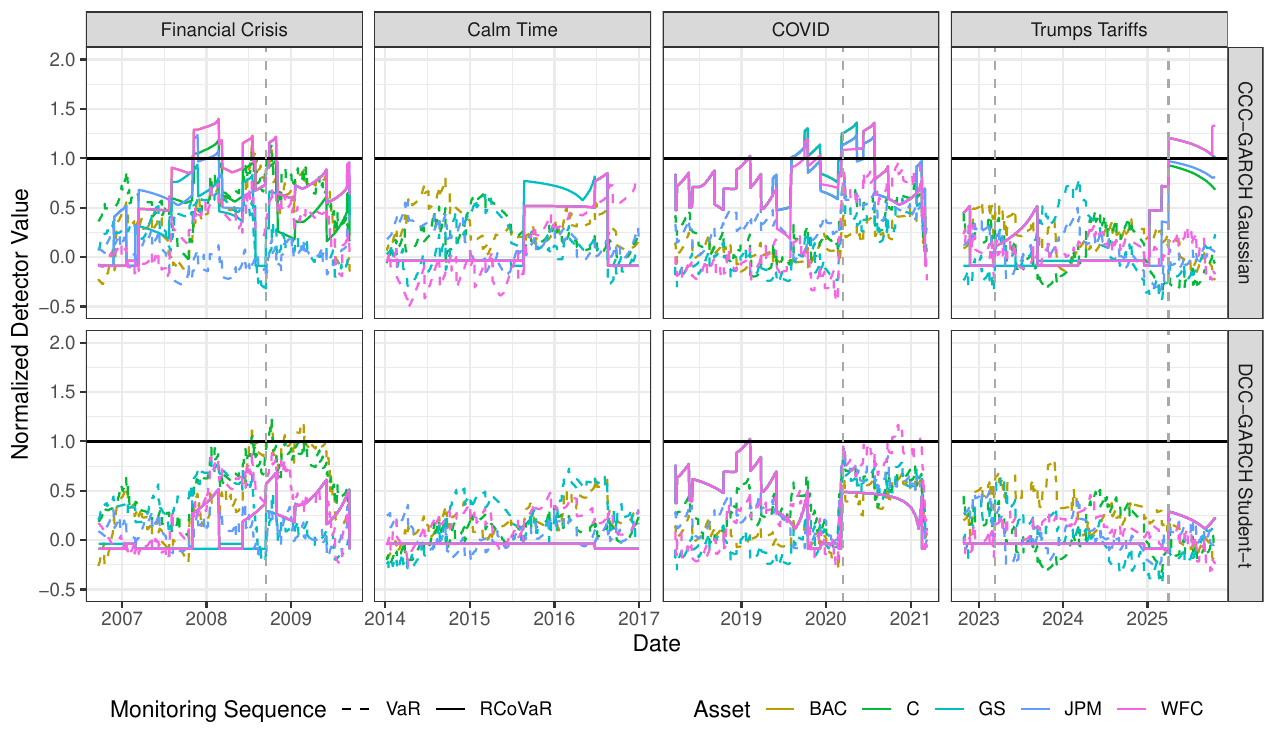}
			\caption{
				Normalized VaR and reverse CoVaR detector values. For details, see the caption of Figure~\ref{fig:Detectors_CoVaR}.
			}
			\label{fig:Detectors_RevCoVaR}
	\end{figure}

		Figure~\ref{fig:Detectors_RevCoVaR} provides corresponding results for the RCoVaR forecasts with overall similar findings.
		While there are no rejections during the calm time in either model, especially the RCoVaR detectors raise alarms for the CCC--GARCH--$\mathcal{N}$ forecasts around the times where (systemic) market risks have increased markedly.
		In contrast, for the Student's $t$ DCC--GARCH forecasts, mostly the VaR detectors yield rejections, implying once more that this model with its dynamic correlation structure is much better suited to model systemic risks.
		% We report similar findings for the CoES and MES in Appendix~\ref{sec:CoESApplication}.
		
		As for the attribution of forecast failure to a specific institution, consider the financial crisis, which had its origins in the housing market \citep{Mis11}.
		During this time, Wells Fargo's increased systemic riskiness was responsible for the forecast failure of the CCC--GARCH--$\mathcal{N}$ model. 
		This holds true when Wells Fargo is viewed both as a systemic risk \textit{receiver} and \textit{transmitter} (see Figures~\ref{fig:Detectors_CoVaR} and \ref{fig:Detectors_RevCoVaR}, respectively).
		The preeminent role of Wells Fargo early on may be explained by its heavy reliance on mortgage lending (later on increased through the acquisition of Wachovia in early 2008), which contrasts with the other considered retail and investment banks.

		Our empirical results are consistent with \citet{LRV12,LRV13}.
		In their horse race of multivariate volatility models, they find that DCC and CCC specifications perform equally well during calm times, yet the former are more adequate for crises periods.
		We find their conclusion for \textit{volatility} forecasting to also hold true for \textit{systemic risk} forecasting. 
		However, in contrast to the (repeated) ``one-shot analysis'' of \citet{LRV12,LRV13}, our monitoring procedures also allow us to date the time point when a forecast breakdown of a specific model occurs and, equally important, also to pinpoint the model failure to a specific institution.

\section{Conclusion}
\label{Conclusion}

Regulators as well as financial institutions are particularly concerned about the commonalities in risk factors, i.e., systemic risk. To effectively take preventive measures, it becomes vital to detect changes in systemic risk assessments as soon as possible. To that end, this paper proposes formal monitoring tools for systemic risk. These are shown to work well in simulations and are useful in practice, as the empirical application to the US~banking sector in Section~\ref{Empirical Application} demonstrates. 

The advantages of our proposed procedures are fourfold.
First, unlike classical ``one-shot'' backtests, our monitoring schemes control size under \emph{repeated application} over a fixed horizon, as required for, e.g., daily risk monitoring in financial markets.
Second, size control holds \emph{in finite samples} by construction, in contrast to asymptotic one-shot backtests such as \citet{FH21}.
Third, our procedures accommodate \emph{multiple time series} simultaneously, unlike the one-shot backtests of \citet{Bea21}, which are restricted to bivariate settings (i.e., $K=1$ in our notation).
Fourth, a Bonferroni-type correction enables us to \emph{attribute} detected deficiencies in systemic risk forecasts to specific institutions, a key feature in regulatory applications where identification matters as much as detection.

We stress that while our empirical application deals with systemic risk in the financial system, our monitoring procedures may be used more widely in other contexts.
For instance, it could be used by individual banks to monitor systemic risk forecasts in their financial positions. 
Again, the key concern of the institutions is not necessarily the risk inherent in their positions (as risk is associated with reward in financial markets), but the commonality in their exposures. 
This is because it is precisely during times of extreme co-movements that diversification benefits vanish, which---as the saying goes---is the only free lunch around. For this reason, it may also be useful for banks to monitor the systemic risk in their positions.

\singlespacing
\small 
\setlength{\bibsep}{4pt}
\bibliographystyle{jaestyle2}
\bibliography{bib_CoVaRMonitoring}

\appendix
\onehalfspacing
\normalsize

\section{Proofs of the main paper}\label{sec:Proofs of the main paper}

\begin{proof}[\textbf{Proof of Proposition~\ref{prop:CoVaR null}:}]
	Since the random variables on both sides of \eqref{eqn:CoVaRH0Law} are binary, their expectations and covariance determine their full probabilistic structure.
	
	First, it holds under $H_0^{\CoVaR}$ that $\E_{t-1} [ I_t] 
		= \E_{t-1} \big[ \1_{\{X_t > \widehat{\VaR}_{t,\beta}\}} \big]
		=  \E_{t-1} \big[ \1_{\{X_t > {\VaR}_{t,\beta}\}} \big]
		= 1- \beta$ for all $t \in \mathbb{N}$.
	Thus, by the law of iterated expectations (LIE), $\E[I_t]=1-\b=\E \big[ \1_{\{U_{1t} > \beta \}} \big]$, such that
	\begin{equation}\label{eq:It}
		I_t\overset{d}{=}\1_{\{U_{1t} > \beta \}},
	\end{equation}
	where $\overset{d}{=}$ denotes equality in distribution.
	
	Second, using the notation $\P_{t-1}\big\{X_t\leq\cdot,Y_{kt}\leq\cdot\big\}:=\P\big\{X_t\leq\cdot,Y_{kt}\leq\cdot\mid\mathcal{F}_{t-1}\big\}$, under $H_0^{\CoVaR}$ and for all $t \in \mathbb{N}$, we have
	\begin{align*}
		\E_{t-1}\big[I_{kt}\big] 
		&= \E_{t-1} \big[\1_{\{X_t> \widehat{\VaR}_{t,\beta},\ Y_{kt}> \widehat{\CoVaR}_{kt,\a|\b}\}} \big] \\
		&= \E_{t-1} \big[\1_{\{X_t> {\VaR}_{t,\beta},\ Y_{kt}> {\CoVaR}_{kt,\a|\b}\}} \big] \\
		&= \P_{t-1}\big\{X_t>\VaR_{t,\beta},\ Y_{kt}>\CoVaR_{kt,\alpha\mid\beta}\big\}\\
		&= \P_{t-1}\big\{X_t>\VaR_{t,\beta}\big\}\P_{t-1}\big\{Y_{kt}>\CoVaR_{kt,\alpha\mid\beta}\mid X_t>\VaR_{t,\beta}\big\}\\
		&= (1-\beta) (1-\alpha).
	\end{align*}
	Therefore, by the LIE, $\E\big[I_{kt}\big]=(1-\beta) (1-\alpha)=\E \big[ \1_{\{U_{1t}>\beta,\ U_{2t}>\alpha\}} \big]$, such that
	\begin{equation}\label{eq:Ikt}
		I_{kt}\overset{d}{=}\1_{\{U_{1t}>\beta,\ U_{2t}>\alpha\}}.
	\end{equation}
	
	Third, by using the above, we get for the contemporaneous covariance that
	\begin{align*}
		\Cov \big(I_t, I_{kt} \big)
		&= \E \big[I_t I_{kt} \big] - \E [I_t] \E [I_{kt}]  \\
		&= \E\Big\{\E_{t-1}\big[\1_{\{X_t> {\VaR}_{t,\beta},\ Y_{kt}> {\CoVaR}_{kt,\a|\b}\}} \big] \Big\}- (1-\beta)^2(1-\alpha) \\
		&= (1-\alpha)(1-\beta) - (1-\beta)^2(1-\alpha) \\
		&= \beta (1-\beta)(1-\alpha)\\
		&= \Big( 1- \E \big[ \1_{\{U_{1t}>\beta\}} \big] \Big) \E \big[ \1_{\{U_{1t}>\beta,\ U_{2t}>\alpha\}} \big] \\
		&=  \E \big[ \1_{\{U_{1t}>\beta\}} \1_{\{U_{1t}>\beta,\ U_{2t}>\alpha\}} \big] - \E \big[ \1_{\{U_{1t}>\beta\}} \big] \E \big[ \1_{\{U_{1t}>\beta,\ U_{2t}>\alpha\}} \big] \\
		&= \Cov \big( \1_{\{U_{1t}>\beta\}}, \1_{\{U_{1t}>\beta,\ U_{2t}>\alpha\}} \big).
	\end{align*}
	In light of this and \eqref{eq:It}--\eqref{eq:Ikt}, we conclude that 
	\begin{equation}\label{eq:Ijoint}
		(I_t, I_{kt})^\prime \overset{d}{=}(\1_{\{U_{1t}>\beta\}}, \1_{\{U_{1t}>\beta,\ U_{2t}>\alpha\}})^\prime.
	\end{equation}
	
	Fourth, to establish independence of $I_s$ and $I_{kt}$ for $s\neq t$, observe that under $H_0^{\CoVaR}$ and for $x,y \in \{0,1\}$,
	\begin{align}
		\P\big\{I_{s}= x\mid\mathcal{F}_{s-1}\big\} &=\beta\1_{\{x=0\}} + (1-\beta)\1_{\{x=1\}},\label{eq:(p.39)}\\
		\P\big\{I_{kt}= y\mid\mathcal{F}_{t-1}\big\}&=\big[1-(1-\alpha)(1-\beta)\big]\1_{\{y=0\}} + (1-\alpha)(1-\beta)\1_{\{y=1\}}.\label{eq:(21+)}
	\end{align}
	Let $s>t$.
	Then, for $x, y \in \{0,1\}$,
	\begin{align*}
		\P\big\{I_{s}= x,\ I_{kt}= y\mid\mathcal{F}_{s-1}\big\}&= \E\big[\1_{\{I_{s}= x,\ I_{kt}= y\}}\mid\mathcal{F}_{s-1}\big]\\
		&= \E\big[\1_{\{I_{kt}= y\}}\1_{\{I_{s}= x\}}\mid\mathcal{F}_{s-1}\big]\\
		&= \1_{\{I_{kt}= y\}}\E\big[\1_{\{I_{s}= x\}}\mid\mathcal{F}_{s-1}\big]\\
		&= \1_{\{I_{kt}= y\}}\P\big\{I_{s}= x\mid\mathcal{F}_{s-1}\big\}\\
		&= \1_{\{I_{kt}= y\}}\big[\beta\1_{\{x=0\}} + (1-\beta)\1_{\{x=1\}}\big]
	\end{align*}
	by \eqref{eq:(p.39)}, such that, by the LIE,
	\begin{align*}
		\P&\big\{I_{s}= x,\ I_{kt}= y\big\}\\
		&= \E\Big[\P\big\{I_{s}= x,\ I_{kt}= y\mid\mathcal{F}_{s-1}\big\}\Big]\\
		&= \E\big[\1_{\{I_{kt}= y\}}\big] \big[\beta\1_{\{x=0\}} + (1-\beta)\1_{\{x=1\}}\big]\\
		&= \Big\{\big[1-(1-\alpha)(1-\beta)\big]\1_{\{y=0\}} + (1-\alpha)(1-\beta)\1_{\{y=1\}}\Big\} \big[\beta\1_{\{x=0\}} + (1-\beta)\1_{\{x=1\}}\big]\\
		&= \P\big\{I_{kt}= y\big\} \P\big\{I_{s}= x\big\}.
	\end{align*}
	Since the case $s<t$ can be treated analogously (using \eqref{eq:(21+)} instead of \eqref{eq:(p.39)}), the above factorization holds for any $s\neq t$.
	Now, showing full independence, i.e.,
	\begin{multline}\label{eq:IND}
		\P\big\{I_{s_1}= x_1,\ldots, I_{s_{\ell}}=x_{\ell},\ I_{kt_1}= y_1,\ldots,I_{kt_m}= y_m\big\}\\
		=\P\big\{I_{s_1}= x_1\big\}\cdot \ldots \cdot\P\big\{I_{s_{\ell}}=x_{\ell}\big\}\cdot \P\big\{I_{kt_1}= y_1\big\}\cdot \ldots \cdot \P\big\{I_{kt_m}= y_m\big\}
	\end{multline}
	for $s_1<\ldots<s_{\ell}$ and $t_1<\ldots<t_m$ ($\ell,m\in\mathbb{N}$) with $s_i\neq t_j$ ($i=1,\ldots,\ell$ and $j=1,\ldots,m$) is only notationally more complicated.
	
	Combining \eqref{eq:It}--\eqref{eq:Ijoint} with \eqref{eq:IND} completes the proof.
\end{proof}

\begin{proof}[\textbf{Proof of Theorem~\ref{thm:CoVaR monitor}:}]
	The actual probability $\iota^{\ast}$ of making a type I error in the sense of \eqref{eqn:CoVaRMonitoringTypeIError}  is
	\begin{align}
		\iota^\ast
		&:=\P_{H_0^{\CoVaR}} \Big\{ \exists T \in \{m,\ldots,n\}: \quad \VaR(T) \ge v  \notag\\
		&\hspace{7cm}\text{  or  }  \;   \CoVaR_{k}(T) \ge c_k \quad \text{for some $k \in [K]$}  \Big\} \notag \\
		& = \P_{H_0^{\CoVaR}}\bigg\{\Big\{\sup_{T=m,\ldots,n}\VaR(T) \geq v\Big\} \cup \bigcup_{k=1}^{K}\Big\{\sup_{T=m,\ldots,n}\CoVaR_{k}(T) \geq c_k\Big\}\bigg\}\notag\\
		& = \P_{H_0^{\CoVaR}}\Big\{\sup_{T=m,\ldots,n}\VaR(T) \geq v\Big\} \notag\\
		&\hspace{1cm} + \P_{H_0^{\CoVaR}}\bigg\{ \bigcup_{k=1}^{K}\Big\{\sup_{T=m,\ldots,n}\CoVaR_{k}(T) \geq c_k\Big\} \; \setminus \; \Big\{\sup_{T=m,\ldots,n}\VaR(T) \geq v\Big\} \bigg\},\label{eq:size}
	\end{align}
	where we used that $\P\big\{A\cup B\big\}=\P\big\{A\big\} + \P\big\{B\setminus A\big\}$ in the third step.

	If $K=1$, all probabilities in the above equation can be computed by virtue of Proposition~\ref{prop:CoVaR null}. However, for $K>1$, the null is silent about the dependence structure between $I_{kt}$ and $I_{k^\prime t}$ ($k\neq k^\prime$) and, hence, between $\CoVaR_{k}(T)$ and $\CoVaR_{k^\prime}(T)$. 
	Similarly as for the well-known Bonferroni correction, we therefore apply Boole's inequality to obtain that
	\begin{align*}
		\P_{H_0^{\CoVaR}}&\bigg\{ \bigcup_{k=1}^{K}\Big\{\sup_{T=m,\ldots,n}\CoVaR_{k}(T) \geq c_k\Big\}\setminus\Big\{\sup_{T=m,\ldots,n}\VaR(T) \geq v\Big\} \bigg\}\\
		&\leq \sum_{k=1}^{K}\P_{H_0^{\CoVaR}}\bigg\{ \Big\{\sup_{T=m,\ldots,n}\CoVaR_{k}(T) \geq c_k\Big\}\setminus\Big\{\sup_{T=m,\ldots,n}\VaR(T) \geq v\Big\} \bigg\}\\
		&= \sum_{k=1}^{K}\P_{H_0^{\CoVaR}}\Big\{\sup_{T=m,\ldots,n}\CoVaR_{k}(T) \geq c_k\Big\}\\
		&\hspace{2cm}- \P_{H_0^{\CoVaR}}\Big\{\sup_{T=m,\ldots,n}\CoVaR_{k}(T) \geq c_k,\ \sup_{T=m,\ldots,n}\VaR(T) \geq v \Big\}.
	\end{align*}
	Plugging this into \eqref{eq:size}, we get that
	\begin{align*}
		\iota^{\ast}&\leq \P_{H_0^{\CoVaR}}\Big\{\sup_{T=m,\ldots,n}\VaR(T) \geq v \Big\} + \sum_{k=1}^{K}\P_{H_0^{\CoVaR}}\Big\{\sup_{T=m,\ldots,n}\CoVaR_{k}(T) \geq c_k\Big\} \\
		&\hspace{2cm} - \sum_{k=1}^{K}\P_{H_0^{\CoVaR}}\Big\{\sup_{T=m,\ldots,n}\CoVaR_{k}(T) \geq c_k,\ \sup_{T=m,\ldots,n}\VaR(T) \geq v \Big\}.
	\end{align*}
	Since the right-hand side equals $\iota$ by \eqref{eq:size ineq}, we can deduce that the actual level $\iota^\ast$ of the monitoring procedure is smaller than or equal to $\iota$.
	Hence, size at level $\iota$ may be controlled \textit{even in finite samples}.
\end{proof}

\begin{proof}[\textbf{Proof of Theorem~\ref{thm:CoVaR rev monitor}:}]
	The actual probability $\iota^{\ast}$ of making a type I error for the rejection rule implied by Theorem~\ref{thm:CoVaR rev monitor} is
	\begin{align*}
		\iota^{\ast} &= \P_{H_0^{\RCoVaR}}\bigg\{\bigcup_{k=1}^{K}\Big\{\sup_{T=m,\ldots,n}\VaR_{k}(T) \geq v_k \Big\}\cup\Big\{\sup_{T=m,\ldots,n} \RCoVaR_{k}(T) \geq c_k \Big\}\bigg\}\\
		&\leq \sum_{k=1}^{K}\P_{H_0^{\RCoVaR}}\bigg\{\Big\{\sup_{T=m,\ldots,n}\VaR_{k}(T) \geq v_k \Big\}\cup\Big\{\sup_{T=m,\ldots,n}\RCoVaR_{k}(T) \geq c_k \Big\}\bigg\}\\
		&=\iota,
	\end{align*}
	where we have used Boole's inequality in the second step, and \eqref{eq:size ineq rev} in the final step.
	Once again, size is controlled.
\end{proof}

\pagebreak 
\onehalfspacing

\section{The CoES and MES}
\label{sex:CoESandMES}

We now turn to the systemic risk measures CoES and MES, which we formally define in Section~\ref{sec:DefCoES}. 
The monitoring procedures are introduced in Section~\ref{sec:CoESMonitoring}, and their finite-sample properties are examined through simulations in Section~\ref{sec:CoESSimulations}.
We defer the proofs of all technical results to Section~\ref{sec:Proofs}.

\subsection{Defining CoES and MES}
\label{sec:DefCoES}

Since CoVaR defined in \eqref{eqn:DefCoVaR} is merely an $\alpha$-quantile of a conditional distribution, it suffers from the same defects as the VaR, discussed, e.g., in \citet{Aea99} and \citet{EKT15}.
In particular, by not considering the magnitude of losses beyond itself, the CoVaR may not adequately capture systemic \textit{tail} risks. 
Therefore, we also consider the CoES, i.e.,
\begin{equation}
	\label{eq:CoES}
	\CoES_{kt,\a|\b}=\frac{1}{1-\a}\int_{\a}^{1}\CoVaR_{kt,\gamma|\b}\D \gamma.
\end{equation}
The CoES also encompasses the MES via $\MES_{kt,\b}=\CoES_{kt,0|\b}$. 
Under our assumption on the CDF of $(X_t,Y_{kt})^\prime\mid\mathcal{F}_{t-1}$ and by using the notation $\E_{t-1}[\cdot]:=\E[\cdot\mid\mathcal{F}_{t-1}]$, we have the intuitive formulas \citep[see, e.g.,][Lemma~2.13]{MFE15}
\begin{align*}
	\CoES_{kt,\a|\b}&=\E_{t-1}\big[Y_t\mid Y_t\geq \CoVaR_{kt,\a|\b},\ X_t\geq\VaR_{t,\b} \big],  \\
	\MES_{kt,\b}&= \E_{t-1}\big[Y_t\mid X_t\geq\VaR_{t,\b}\big].
\end{align*}

\subsection{CoES and MES Monitoring}
\label{sec:CoESMonitoring}

We now propose monitoring procedures for the CoES as defined in \eqref{eq:CoES} and for the MES, which arises as a special case of the CoES for $\alpha = 0$.
Recall that our monitoring approach for the CoVaR in Sections~\ref{CoVaR Monitoring}--\ref{ReversedCoVaRMonitoring} crucially relies on the fact that conditional calibration implies a given probabilistic behavior of the \emph{binary} sequences $I_t$ and $I_{kt}$ in \eqref{eqn:Indicators}.
This results from the VaR and CoVaR being certain quantiles, such that the associated identification functions in \eqref{eqn:CoVaRJointID} are binary.
In contrast, for the CoES (and MES), the respective joint identification functions given in equations~(S.5)--(S.6) in the supplementary material of \citet{FH21} take continuous values and their full probabilistic structure under the null is not known.
Hence, a direct extension of the CoVaR monitoring procedure with sequential finite-sample guarantees as in \eqref{eqn:MonitoringGuarantee} is not possible.
Therefore, we do not directly test the null of ideal VaR and CoES forecasts.

Instead, we test that the forecasts $\widehat{\VaR}_{t,\beta}$ and $\widehat{\CoVaR}_{kt,\gamma|\beta}$ are ideal for all $\gamma\in[\alpha,1)$.
This is closely related to testing VaR and CoES calibration, because $\CoES_{kt,\a|\b}=\frac{1}{1-\a}\int_{\a}^{1}\CoVaR_{kt,\gamma|\b}\D \gamma$ is defined by an integral over the CoVaR for $\gamma\in[\alpha,1)$.
Formally:
\begin{multline*}
	H_0^{\CoES}: \quad 
	\widehat{\VaR}_{t,\beta} = \VaR_{t,\b}\quad
	\text{and}\quad
	\widehat{\CoVaR}_{kt,\gamma|\beta} = \CoVaR_{kt,\gamma\mid\beta}\\
	\text{ for all }\gamma \in [\alpha, 1),\ k \in [K], \text{ and }t\in \mathbb{N}.
\end{multline*}

To test the ``CoVaR part'' of this null, we use the notation $I_{kt,\gamma|\b} := \1_{\{X_t> \widehat{\VaR}_{t,\beta},\ Y_{kt}> \widehat{\CoVaR}_{kt,\gamma|\b}\}}$ and follow \citet{Bea21} in considering the \textit{cumulative CoVaR violation sequence}

\begin{equation}
	\label{eq:CCVP}
	H_{kt,\a|\b} 
	:= \frac{1}{1-\a}\int_{\a}^{1} I_{kt,\gamma|\b} \D \gamma 
	=	\frac{\1_{\{X_t> \widehat{\VaR}_{t,\beta}\}}}{1-\a}\int_{\a}^{1} \1_{\{Y_{kt}> \widehat{\CoVaR}_{kt,\gamma|\b}\}} \D \gamma.
\end{equation}
Such cumulative (integrated) quantile violations in \eqref{eq:CCVP} are well-known from backtesting the univariate risk measure ES; see, e.g., \citet{Acerbi2002spectral}, \citet{DE17} and \citet{Du2024powerful}.
While the classical cumulative violation sequence integrates over quantile violations $I_t$, our conditional version in \eqref{eq:CCVP} uses the CoVaR-version of the violations $I_{kt,\gamma|\b}$.
Importantly, the sequence in \eqref{eq:CCVP} as an integral over CoVaR violations resembles the definition of the CoES in \eqref{eq:CoES} as an integral over the CoVaR values and, therefore, provides a suitable tool for CoES monitoring.

Our cumulative CoVaR violation sequence is related to a conditional tail-version of the classical \textit{probability integral transformations, PITs}, (also called \citet{Ros52} transformations) of $(X_t, Y_{kt})^\prime$.
These are an essential tool for assessing calibration of distributional forecasts \citep{DGT98, DM20, GR_2023}.
Define the classical PIT of $X_t$ and the conditional tail version of $Y_{kt}\mid X_t>\VaR_{t,\b}$ as
\begin{align}
	\label{eqn:PITdef}
	\widehat{U}_t^{X} = \widehat{F}_{X_t\mid\mathcal{F}_{t-1}}(X_t)
	\qquad  \text{ and } \qquad 
	\widehat{U}_t^{Y_k\mid X>\VaR} = \widehat{F}_{Y_{kt}\mid X_t>\VaR_{t,\b},\mathcal{F}_{t-1}}(Y_{kt}),
\end{align}
with ``population'' counterparts $U_t^{X} = F_{X_t\mid\mathcal{F}_{t-1}}(X_t)$ and $U_t^{Y_k\mid X>\VaR} = F_{Y_{kt}\mid X_t>\VaR_{t,\b},\mathcal{F}_{t-1}}(Y_{kt})$.
Now, our cumulative CoVaR violation sequence is related to the \emph{conditional tail PIT}
\begin{align}
	\label{eqn:CumViolPITs}
	\widetilde{H}_{kt,\a|\b} := \1_{\big\{\widehat{U}_t^{X}> \b,\ \widehat{U}_t^{Y_k\mid X>\VaR}>\a\big\}}\frac{\widehat{U}_t^{Y_k\mid X>\VaR}-\a}{1-\a}.
\end{align}
$\widetilde{H}_{kt,\a|\b}$ in \eqref{eqn:CumViolPITs} is \textit{conditional} as it involves the conditional PIT, $\widehat{U}_t^{Y_k\mid X>\VaR}$, and concerns the \textit{tail} only, as it is truncated through the indicator function.
The precise relation to $H_{kt,\a|\b}$ is given in the following proposition.

\begin{prop} 
	\label{rem:CumViolSeq}
	Suppose that $\beta \mapsto \widehat{\VaR}_{t,\b}=\widehat{F}^{-1}_{X_t\mid\mathcal{F}_{t-1}}(\b)$ and $\gamma \mapsto \widehat{\CoVaR}_{kt,\gamma|\b} = \widehat{F}^{-1}_{Y_{kt}\mid X_t>\VaR_{t,\b},\mathcal{F}_{t-1}}(\gamma)$ are continuous and strictly increasing functions in a neighbourhood around $\beta$, and for an open set containing all $\gamma \in (\alpha, 1)$, respectively.
	Then, the conditional tail PIT from \eqref{eqn:CumViolPITs} equals the cumulative CoVaR violation sequence from \eqref{eq:CCVP}, i.e., $H_{kt,\a|\b}=\widetilde{H}_{kt,\a|\b}$.
\end{prop}

Proposition~\ref{rem:CumViolSeq} is useful to derive several (testable) properties of $H_{kt,\a|\b}$ under 	$H_0^{\CoES}$, viz.~IIDness of $H_{kt,\a|\b}$ and its CDF.
We formally do so in the following proposition.

\begin{prop}
	\label{prop:CumViolProcess}
	Under $H_0^{\CoES}$, it holds for all $k \in [K]$ that
	\[
	\big\{(I_t,\, H_{kt,\alpha\mid\beta})^\prime\big\}_{t\in\mathbb{N}}\overset{d}{=}\bigg\{\Big(\1_{\{U_{1t}>\b\}},\, \1_{\{U_{1t}>\b,\ U_{2t}>\a\}}\frac{U_{2t}-\a}{1-\a}\Big)^\prime\bigg\}_{t\in\mathbb{N}},
	\]
	where $U_{it}\overset{IID}{\sim}\mathcal{U}[0,1]$ are independent of each other for  $i=1,2$ and all $t \in \mathbb{N}$.
\end{prop}

Under $H_0^{\CoES}$, Proposition~\ref{prop:CumViolProcess} implies that 
\begin{align}
	\label{eqn:H_CDF}
	H(x):=\P\big\{H_{kt,\a|\b}\leq x\big\} = \big[x(1-\a) + \a\big](1-\b) + \b,\qquad x\in[0,1],
\end{align}
and $H(x) = 0$ for any $x < 0$.
The two properties of $H_{kt,\a|\b}$---its IIDness and the specific form of its CDF---allow for a feasible way of monitoring $H_{0}^{\CoES}$.
Notice that as argued in Remark~\ref{rem:VaRMonitoring} for the CoVaR, we simultaneously have to monitor the condition $\widehat{\VaR}_{t,\beta} = \VaR_{t,\b}$ in $H_0^{\CoES}$, which we do by employing the VaR detector from \eqref{eq:VaR det}.

The unconditional part of our CoES detector monitors whether the CDF of $H_{kt,\a|\b}$ equals the uniform-type CDF $H(\cdot)$ by sequentially computing the Kolmogorov--Smirnov statistic $D_{kT} = \sup_{x\in[0,1]}\big|F_{kT}(x)-H(x)\big|$, where $F_{kT}(\cdot)$ denotes the empirical CDF of $\{H_{kt,\a|\b}\}_{t=T-m+1,\ldots,T}$.
Based on these considerations, we propose the detector
\begin{align*}
	\CoES^{uc}_{k}(T) = \frac{D_{kT} - \E_{H_0^{\CoES}}[D_{kT}] }{ \sqrt{\Var_{H_0^{\CoES}}(D_{kT}) }},
\end{align*}
where $\E_{H_0^{\CoES}}[D_{kT}]$ and $\Var_{H_0^{\CoES}}(D_{kT})$  denote the null-hypothetical mean and variance of $D_{kT}$.

It remains to check IIDness of $\{H_{kt,\a|\b}\}_{t \in [n]}$.
For this, we draw on \citet{Hon96}, who proposes to assess the non-autocorrelatedness with a spectral density-based test statistic
\begin{align*}
	\CoES^{iid}_{k}(T) = \frac{M_{kT} - \E_{H_0^{\CoES}}[M_{kT}] }{ \sqrt{\Var_{H_0^{\CoES}}(M_{kT}) }},
	\qquad \text{where} \qquad 
	M_{kT} = m\sum_{j=1}^{m-1} \kappa^{2}(j/p_m) \, \widehat{\rho}_{kj}^2(T).
\end{align*}
Here, $\widehat{\rho}_{kj}(T) = \widehat{\gamma}_{kj} (T)/\widehat{\gamma}_{k0}(T)$ is the $j$-th  sample autocorrelation of the $k$-th series, where
\[
\widehat{\gamma}_{kj}(T)=m^{-1}\sum_{t=T-m+1+|j|}^{T}\big[H_{kt,\a|\b} - \overline{H}_{k,\a|\b}(T)\big] \big[H_{k,t-|j|,\a|\b} - \overline{H}_{k,\a|\b}(T)\big]
\]
is the $j$-th sample autocovariance of $\{H_{kt,\a|\b}\}_{t=T-m+1,\ldots,T}$, with $\overline{H}_{k,\a|\b}(T)=m^{-1}\sum_{t=T-m+1}^{T}H_{kt,\a|\b}$ denoting the sample mean.
Furthermore, $\kappa:\mathbb{R}\to[-1,1]$ is some kernel function and $p_m$ denotes a smoothing parameter satisfying $p_m\longrightarrow\infty$ and $p_m/m\longrightarrow0$, as $m\to\infty$.
While the choice $\kappa(z)=\1_{\{|z|\leq1\}}$ recovers the classical Box--Pierce test, we follow the power considerations of \citet[Theorem~5]{Hon96}  and use the Daniell kernel $\kappa(z)=\sin(\pi z)/[\pi z]$, $z\in\mathbb{R}$.  
We employ a ``small'' choice  $p_m= \log(m)$ as \citet{Hon96} reports that small $p_m$ lead to higher power, though possibly at the expense of some size distortions.
Yet, these potential size distortions are of no concern for our monitoring, where size is controlled in \textit{finite} samples.

Of course, many other detectors (test statistics) could potentially be used to detect deviations from independence \citep[see][for an overview]{Hon10}.
We opt for $M_T$ here because it is intuitive, computationally easy to handle and powerful, as shown by \citet{Hon96}.

Putting together the individual components, we obtain the CoES detector
\[
\CoES_{k}(T)=a\CoES^{uc}_{k}(T) + (1-a)\CoES^{iid}_{k}(T),\qquad a\in[0,1].
\]
As for the CoVaR, we simultaneously monitor calibration of the associated VaR forecasts via \eqref{eq:VaR det}.
The validity of the CoES monitoring procedure is established by the following theorem, which corresponds verbatim to Theorem~\ref{thm:CoVaR monitor}, with $\CoVaR_{k}(T)$ replaced by $\CoES_{k}(T)$.

\begin{thm}
	\label{thm:CoES monitor}
	For any $\iota \in (0,1)$, it holds that
	\begin{multline}
		\label{eqn:CoESMonitoringTypeIError}
		\P_{H_0^{\CoES}} \Big\{ \exists T \in \{m,\ldots,n\}: \quad \VaR(T) \ge v  \; \\
		\text{  or  }  \;   \CoES_{k}(T) \ge c_k \quad \text{for some $k \in [K]$}  \Big\} \le \iota,
	\end{multline}
	if the critical values $v$ and the $c_k$'s are chosen such, that
	\begin{multline}\label{eq:size ineq_CoES}	
		\P_{H_0^{\CoES}}\Big\{\sup_{T=m,\ldots,n}\VaR(T) \geq v \Big\} + \sum_{k=1}^{K}\P_{H_0^{\CoES}}\Big\{\sup_{T=m,\ldots,n}\CoES_{k}(T) \geq c_k\Big\} \\
		- \sum_{k=1}^{K}\P_{H_0^{\CoES}}\Big\{\sup_{T=m,\ldots,n}\VaR(T) \geq v,\ \sup_{T=m,\ldots,n}\CoES_{k}(T) \geq c_k \Big\}=\iota.
	\end{multline}
\end{thm}

\begin{proof}
Analogous to that of Theorem~\ref{thm:CoVaR monitor}, where $\CoES_{k}(T)$ simply replaces $\CoVaR_{k}(T)$ at every occurrence.
\end{proof}

As discussed after Theorem~\ref{thm:CoVaR monitor} for the CoVaR monitoring, we advocate choosing $v$ and $c_k$ to correspond to some $(1-\nu)$-quantiles of the null distributions of $\sup_{T=m,\ldots,n}\VaR(T)$ and $\sup_{T=m,\ldots,n}\CoES_{k}(T)$, respectively. 
These critical values (denoted $v^{\CoES}$ and $c^{\CoES}$) are computed by following Algorithm~\ref{algo:2} in Appendix~\ref{sec:Algorithms}.

By taking $\a=0$ in Theorem~\ref{thm:CoES monitor}, we obtain a monitoring scheme for the MES, since $\MES_{kt,\b}=\CoES_{kt,0|\b}$. 
In that case, we write the CoES detector $\CoES_{k}(T)$ simply as $\MES_{k}(T)$ with appertaining  critical values $v^{\MES}$ and $c^{\MES}$.

\begin{rem}[Conditional tail monitoring]
	\label{rem:CondTailMonitoring}
	Formally, our monitoring procedure for the CoES requires conditional tail forecasts to compute \eqref{eqn:PITdef}--\eqref{eqn:CumViolPITs}.
	In the strict sense of \citet[Definition~1.1]{BD22}, this implies that we essentially monitor the conditional tail distribution beyond the CoVaR.
	However, as argued by \citet{GordyMcNeil2020} and \citet{Hue2024backtesting}, (tail) PITs are a ``middle ground'' between reporting risk measures and (tail) distributions and their reporting is already mandatory for the US banking system; see \citet[p.~53105]{FederalRegister2012}:
	In detail, while reporting full (tail) distributions would allow to reach conclusions about the banks' internal risk models, reporting of (tail) PITs is sufficient for many evaluation metrics, yet at the same time conceals the confidential tail distribution forecasts.
	While forecasts for the tail distribution can and should of course be submitted before the return materializes, the PITs also rely on the corresponding observation and, hence, are only available \emph{ex post}. 
	A fair PIT reporting mechanism could be set up by submitting the tail distributions to a ``trustee'' before the return materializes, and the trustee then only reports the materialized PIT value (according to the observed return).
	Such a procedure would allow for a proper tail forecast evaluation, while maintaining confidentiality of the full predictive (tail) distributions.
\end{rem}

\subsection{Simulations for the CoES and MES}
\label{sec:CoESSimulations}

\begin{table}[tb]
	\centering
	\resizebox{\linewidth}{!}{
		\small
		\begin{tabular}{ccccc ccccc ccccc ccccc ccccc} 
			\toprule
			Measure && $\beta$ &  & $K$ &  & Joint &  & VaR & &  \multicolumn{10}{c}{CoES/MES for series number $k$} \\
			\cmidrule{11-20}  
			&&&&&&&&&& 1 & 2 & 3 & 4 & 5 & 6 & 7 & 8 & 9 & 10 \\ 
			\midrule
			\multirow{8.7}{*}{\shortstack{CoES: \\ $\alpha = \beta$}} && \multirow{4}{*}{0.9} &  & 1 &  & 10.56 &  & 5.28 && 5.28 &  &  &  &  &  &  &  &  &  \\ 
			&&  &  & 2 &  & 8.42 &  & 3.24 && 2.54 & 2.68 &  &  &  &  &  &  &  &  \\ 
			&& &  & 5 &  & 8.16 &  & 1.54 && 1.38 & 1.10 & 1.56 & 1.30 & 1.40 &  &  &  &  &  \\ 
			&&  &  & 10 &  & 8.70 &  & 0.82 && 0.88 & 0.64 & 0.78 & 0.86 & 0.84 & 1.04 & 1.02 & 0.82 & 0.74 & 0.56 \\ 
			\cmidrule{3-20}   
			&& \multirow{4}{*}{0.95} &  & 1 &  & 9.58 &  & 4.32 && 5.26 &  &  &  &  &  &  &  &  &  \\ 
			&&  &  & 2 &  & 9.94 &  & 3.48 && 3.22 & 3.36 &  &  &  &  &  &  &  &  \\ 
			&& &  & 5 &  & 9.38 &  & 1.44 && 1.86 & 1.74 & 1.76 & 2.02 & 1.68 &  &  &  &  &  \\ 
			&&  &  & 10 &  & 8.00 &  & 0.72 && 0.78 & 0.78 & 0.70 & 0.90 & 1.06 & 0.70 & 0.82 & 1.06 & 1.16 & 0.88 \\ 
			\midrule 
			\multirow{8.7}{*}{\shortstack{MES: \\ $\alpha = 0$}} && \multirow{4}{*}{0.9} &  & 1 &  & 10.60 &  & 5.52 && 5.12 &  &  &  &  &  &  &  &  &  \\ 
			&&  &  & 2 &  & 9.04 &  & 3.16 && 3.12 & 2.96 &  &  &  &  &  &  &  &  \\ 
			&& &  & 5 &  & 8.18 &  & 1.44 && 1.50 & 1.44 & 1.76 & 1.12 & 1.24 &  &  &  &  &  \\ 
			&& &  & 10 &  & 5.16 &  & 0.76 && 0.36 & 0.48 & 0.54 & 0.38 & 0.48 & 0.38 & 0.60 & 0.48 & 0.48 & 0.56 \\ 
			\cmidrule{3-20}  
			&& \multirow{4}{*}{0.95}&  & 1 &  & 9.06 &  & 4.24 && 4.82 &  &  &  &  &  &  &  &  &  \\ 
			&&  &  & 2 &  & 10.44 &  & 3.52 && 3.82 & 3.42 &  &  &  &  &  &  &  &  \\ 
			&& &  & 5 &  & 7.26 &  & 1.62 && 1.42 & 1.38 & 1.12 & 1.32 & 1.36 &  &  &  &  &  \\ 
			&&  &  & 10 &  & 5.88 &  & 0.68 && 0.48 & 0.64 & 0.62 & 0.60 & 0.60 & 0.66 & 0.72 & 0.68 & 0.64 & 0.66 \\ 
			\bottomrule
		\end{tabular}
	}
	\caption{Joint and disaggregated \textit{first} rejection rates (in percent) of our CoES and MES surveillance procedure under the null hypothesis of correctly specified forecasts. 
		Results are reported for a nominal level of $\iota = 10\%$, for $K \in \{1,2,5,10\}$ and $\beta \in \{0.9, 0.95\}$, with $\alpha = \beta$ for the CoES, and $\alpha = 0$ for the MES.
		The column ``VaR'' reports the percentage of cases in which the VaR detector raised the first alarm, while the columns labeled ``1''--``10'' report the corresponding first-alarm rates of the respective CoVaR detectors.}
	\label{tab:SizeCoES}
\end{table}

\begin{figure}[tb]
	\centering
	\includegraphics[width=\textwidth]{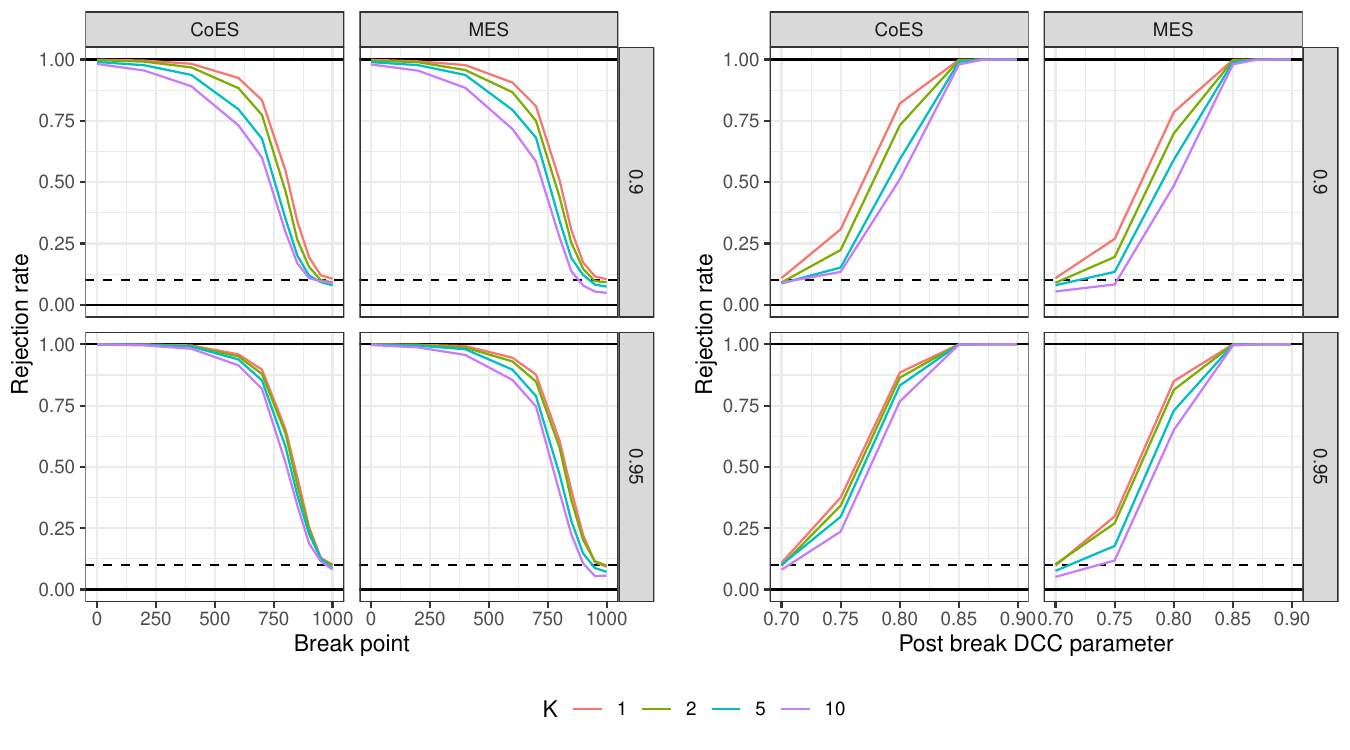}
	\caption{Rejection rates of our CoES and MES surveillance methods plotted against the break point in \eqref{eqn:SimParamMisspec} in the left panel and against the post-break DCC parameter in the right panel.
	In both plots, we consider $\alpha = \beta \in \{0.9, 0.95\}$ and $K \in \{1,2,5,10\}$.}
	\label{fig:PowerCoES}
\end{figure}

\begin{figure}[tb]
	\centering
	\includegraphics[width=0.75\textwidth]{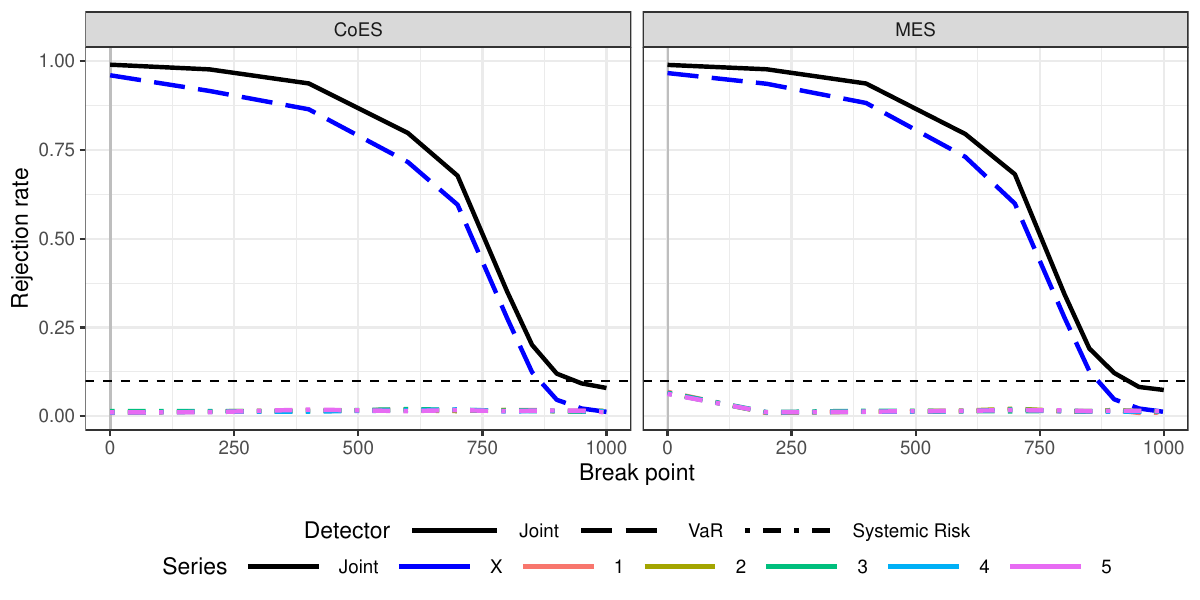}
	\caption{Rejection rates of the joint CoES and MES procedure in black, and frequencies how often the individual detectors are the \emph{first} to generate a detection in colors in the setting of the left panel of Figure~\ref{fig:PowerCoVaR}.
		Dashed lines depict ``first'' rejection rates for the VaR detector, and dot-dashed lines for the systemic risk detectors. The colors indicate the respective component of $\mW_t$. 
		The nominal level of $\iota = 10\%$ is indicated by the dashed horizontal line.}
	\label{fig:PowerCoESIndividual}
\end{figure}

\begin{figure}[tb]
	\centering
	\includegraphics[width=\textwidth]{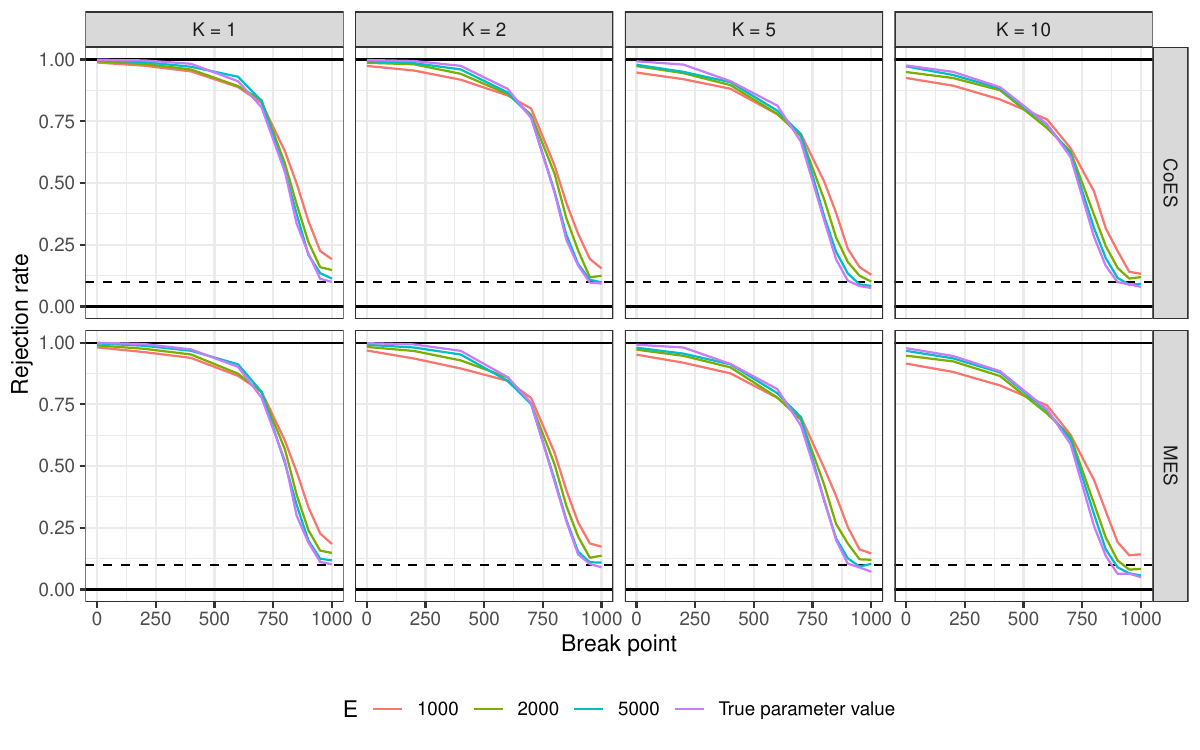}
	\caption{Rejection rates of our CoES and MES surveillance methods plotted against the break point for different estimation window lengths $E$, where $E = \infty$ implies the use of the correct parameters.
		We consider $K \in \{1,2,5,10\}$, and further keep $n=1000$, $t^\ast = 0$,  $\beta_\text{post} = 0.85$ and $\alpha = \beta = 0.9$ fixed.}
	\label{fig:PowerEstWindowCoES}
\end{figure}

Here, we extend the simulations from Section~\ref{Simulations} to the CoES and MES.
We use the same DGP as in Section~\ref{Data-Generating Process} and generate covariance matrix forecasts as described in Section~\ref{sec:SystemicRiskForecasts}.
Based on the implied forecasted Student's $t$ distribution (nesting the Gaussian special case for $\nu = \infty$), we obtain the CoES (and MES) sequence $H_{kt,\a|\b}$ by noting that
\begin{align*}
	H_{kt,\a|\b} = \frac{\1_{\{X_t> \widehat{\VaR}_{t,\beta}\}}}{1-\a}\int_{\a}^{1} \1_{\{Y_{kt}> \widehat{\CoVaR}_{kt,\gamma|\b}\}} \D \gamma 
	= \frac{\1_{\{X_t> \widehat{\VaR}_{t,\beta}\}}}{1-\a}  \1_{\{\xi_{kt}(Y_{kt}) \ge \alpha\}} \big( \xi_{kt}(Y_{kt}) - \alpha \big),
\end{align*}
where $\xi_{kt}(y_{kt})$ is obtained for any realization $y_{kt}$ of the random variable $Y_{kt}$ as
\begin{align*}
	\xi_{kt}(y_{kt}) = \mathbb{P} \big\{ \widetilde{X}_t > \widehat{\VaR}_{t,\b}, \widetilde{Y}_{kt} > y_{kt} \; \big| \; \mathcal{F}_{t-1} \big\} \big/ (1-\beta),
\end{align*}
where $(\widetilde{X}_t, \widetilde{Y}_{kt})^\prime \mid \mathcal{F}_{t-1} \sim t_\nu(\widehat{\mH}_{t,k})$ follow the same law as $(X_t, Y_{kt})^\prime  \mid \mathcal{F}_{t-1}$.

Table~\ref{tab:SizeCoES} displays results as in Table~\ref{tab:SizeCoVaR}, but for the CoES and MES instead of the CoVaR.
The results are almost identical: ``Joint'' size is held exactly for $K=1$ and slightly conservatively for $K \in \{2,5,10\}$ with an equal distribution among the VaR and CoES/MES detectors.

The analysis of our CoES and MES monitoring procedures' power in Figures~\ref{fig:PowerCoES}--\ref{fig:PowerEstWindowCoES} follows Figures~\ref{fig:PowerCoVaR}--\ref{fig:PowerEstWindowCoVaR}, but displays rejection rates of the CoES and MES instead of for the CoVaR and RCoVaR.
The results are again very similar to the CoVaR: Figure~\ref{fig:PowerCoES} displays good power under misspecification, and Figure~\ref{fig:PowerEstWindowCoES} a similar behavior under parameter estimation noise.
The rates when the CoES/MES detectors reject first in Figure~\ref{fig:PowerCoESIndividual}, however, is considerably smaller for the CoES and MES than for the CoVaR.
This might be caused by the more complicated construction of the detectors for the CoES than for the CoVaR due to the non-binary CoES identification functions.
Notice however that Figure~\ref{fig:PowerCoESIndividual} displays which detector rejects \emph{first}, and not whether a given detector rejects \emph{at all}.

\FloatBarrier

\section{Algorithms}
\label{sec:Algorithms}

Here, we describe how critical values for the Reverse CoVaR, the CoES and MES are computed by modifying Algorithm~\ref{algo:1}.

For the Reverse CoVaR and Theorem~\ref{thm:CoVaR rev monitor}, we compute the critical values as follows:
\begin{algorithm}
	\label{algo:RevCoVaR}
	To compute critical values for Theorem~\ref{thm:CoVaR rev monitor} proceed as follows:
	\begin{enumerate}
		\item
		\label{it:RevCoVaR1} 
		Generate a large number $B$ of mutually independent samples $\big\{U_{1t}\overset{\text{IID}}{\sim}\mathcal{U}[0,1]\big\}_{t \in [n]}$ and $\big\{U_{2t}\overset{\text{IID}}{\sim}\mathcal{U}[0,1]\big\}_{t \in [n]}$.

		\item
		\label{it:RevCoVaR2} 
		\sloppy
		Compute the $B$ sequences $\big\{I_{t}^\ast=\1_{\{U_{1t} > \b\}}\big\}_{t \in [n]}$ and $\big\{I_{1t}^{\ast}=\1_{\{U_{1t} > \b,\ U_{2t}>\a \}}\big\}_{t \in [n]}$.
		
		\item
		\label{it:RevCoVaR3} 
		For $b=1,\ldots,B$, calculate:
		\begin{itemize}
			\item[i.] 
			$\sup_{T=m,\ldots,n} \VaR^{b}(T)$, where $\VaR^{b}(T)$ is defined as $\VaR(T)$, except that $\{I_t\}$ is replaced by the $b$-th sample from $\{I_t^{\ast}\}$ from step \ref{it:2} of this algorithm;
			
			\item[ii.] 
			$\sup_{T=m,\ldots,n}\RCoVaR^{b}(T)$, where $\RCoVaR^{b}(T)$ is defined as $\RCoVaR_{1}(T)$, except that $\{I_{1t}\}$ is replaced by the $b$-th sample from $\{I_{1t}^{\ast}\}$ from step \ref{it:2} of this algorithm.
		\end{itemize}
		
		\item 
		On a fine grid for $\nu\in[0,1]$, compute the empirical $(1-\nu)$-quantiles of 
		\begin{itemize}
			\item[i.] $\big\{\sup_{T=m,\ldots,n} \VaR^{b}(T)\big\}_{b \in [B]}$, which we denote by $v(\nu)$, and 
			\item[ii.] $\big\{\sup_{T=m,\ldots,n} \RCoVaR^{b}(T)\big\}_{b \in [B]}$, which we denote by $c(\nu)$.
		\end{itemize}
		
		\item 
		Find the value $\nu \in [0,1]$ for which
		\begin{align*}
			  \frac{K}{B}\sum_{b=1}^{B} \1_{\big\{\sup_{T=m,\ldots,n}\VaR^{b}(T)\geq v(\nu) \; \vee \;\sup_{T=m,\ldots,n} \RCoVaR^{b}(T)\geq c(\nu) \big\}}
		\end{align*}
		is closest to (or smaller than) $\iota$; cf.~\eqref{eq:size ineq rev}. 
		The appertaining critical values will be denoted by
		$v^{\RCoVaR}$ and $c^{\RCoVaR}$.
	\end{enumerate}
\end{algorithm}

\begin{algorithm}
	\label{algo:2}
	To compute critical values for Theorem~\ref{thm:CoES monitor} proceed as follows:
	\begin{enumerate}
		\item
		\label{it:21} 
		Generate a large number $B$ of mutually independent samples $\big\{U_{1t}\overset{\text{IID}}{\sim}\mathcal{U}[0,1]\big\}_{t \in [n]}$ and $\big\{U_{2t}\overset{\text{IID}}{\sim}\mathcal{U}[0,1]\big\}_{t \in [n]}$.

		\item 
		\label{it:22} 
		Compute the $B$ sequences $\big\{I_{t}^\ast=\1_{\{U_{1t} > \b\}}\big\}_{t \in [n]}$ and $\big\{H_{1t,\a|\b}^{\ast}=\1_{\{U_{1t} > \b,\ U_{2t}>\a \}}(U_{2t}-\a)/(1-\a)\big\}_{t \in [n]}$.
		
		\item
		\label{it:23} 
		For $b=1,\ldots,B$, calculate:
		\begin{itemize}
			\item[i.] 
			$\sup_{T=m,\ldots,n}\VaR^{b}(T)$, where $\VaR^{b}(T)$ is defined as $\VaR(T)$, except that $\{I_t\}$ is replaced by the $b$-th sample from $\{I_t^{\ast}\}$ from step \ref{it:22} of this algorithm;
			
			\item[ii.] 
			$\sup_{T=m,\ldots,n}\CoES^{b}(T)$, where $\CoES^{b}(T)$ is defined as $\CoES_{1}(T)$, except that $\{H_{1t,\a|\b}\}$ is replaced by the $b$-th sample from $\{H_{1t,\a|\b}^{\ast}\}$ from step \ref{it:22} of this algorithm.
		\end{itemize}
		
		\item On a fine grid for $\nu\in[0,1]$, compute the empirical $(1-\nu)$-quantiles of 
		\begin{itemize}
			\item[i.] $\big\{\sup_{T=m,\ldots,n}\VaR^{b}(T)\big\}_{b \in [B]}$, which we denote by $v(\nu)$ and 
			\item[ii.] $\big\{\sup_{T=m,\ldots,n}\CoES^{b}(T)\big\}_{b \in [B]}$, which we denote by $c(\nu)$.
		\end{itemize}
		
		\item Find the value $\nu \in [0,1]$ for which
		\begin{multline*}
			\frac{1}{B}\sum_{b=1}^{B}\1_{\big\{\sup_{T=m,\ldots,n}\VaR^{b}(T)\geq v(\nu) \big\}} + \frac{K}{B}\sum_{b=1}^{B}\1_{\big\{\sup_{T=m,\ldots,n}\CoES^{b}(T)\geq c(\nu) \big\}} \\
			- \frac{K}{B}\sum_{b=1}^{B} \1_{\big\{\sup_{T=m,\ldots,n} \VaR^{b}(T)\geq v(\nu),\ \sup_{T=m,\ldots,n} \CoES^{b}(T)\geq c(\nu) \big\}}
		\end{multline*}
		is closest to (or smaller than) $\iota$; cf.~\eqref{eq:size ineq_CoES}.
		The appertaining critical values will be denoted by $v^{\CoES}$ and $c^{\CoES}$.
	\end{enumerate}
\end{algorithm}

As above, the validity of the resulting critical values from Algorithm~\ref{algo:2} relies on the fact that the sequence $\big\{(I_t^\ast,H_{1t,\a|\b}^{\ast})^\prime\big\}_{t\in[n]}$ has the same probabilistic properties as $\big\{(I_t, H_{kt,\a|\b})^\prime\big\}_{t\in[n]}$ has under the null.
Hence, the critical values $v^{\CoES}$ and $c^{\CoES}$ from Algorithm~\ref{algo:2} can be computed with arbitrary precision by choosing $B$ sufficiently large.

\section{Proofs for CoES and MES Monitoring}
\label{sec:Proofs}

\begin{proof}[\textbf{Proof of Proposition~\ref{rem:CumViolSeq}:}]
	The proof follows from the following calculation, where we write \eqref{eq:CCVP} as
	\begin{align*}
		H_{kt,\a|\b} 
		&= 	\frac{\1_{\{X_t> \widehat{\VaR}_{t,\beta}\}}}{1-\a}\int_{\a}^{1} \1_{\{Y_{kt}> \widehat{\CoVaR}_{kt,\gamma|\b}\}} \D \gamma  \\
		&= \frac{\1_{\big\{X_t > \widehat{F}_{X_t\mid\mathcal{F}_{t-1}}^{-1}(\b)\big\}}}{1-\a}
		\int_{\a}^{1} \1_{ 		\big\{ Y_{kt} > \widehat{F}^{-1}_{ Y_{kt} \mid X_t \geq \VaR_{t,\b}, \mathcal{F}_{t-1} }(\gamma) \big\} }  \D \gamma  	\\
		&= \frac{\1_{\big\{ \widehat{F}_{X_t\mid\mathcal{F}_{t-1}}(X_t) > \b  \big\}}}{1-\a}	\int_{\a}^{1}\1_{\big\{ \widehat{F}_{Y_{kt}\mid X_t>\VaR_{t,\beta},\mathcal{F}_{t-1}}(Y_{kt})> \gamma \big\}} \D \gamma\\
		&= \frac{\1_{\big\{\widehat{U}_t^{X}> \b\big\}}}{1-\a} \int_{\a}^{1}\1_{\big\{\widehat{U}_t^{Y_k\mid X>\VaR}> \gamma\big\}} \D \gamma\\
		&= \frac{\1_{\big\{\widehat{U}_t^{X}> \b\big\}}}{1-\a} \1_{\big\{\widehat{U}_t^{Y_k\mid X>\VaR}>\a\big\}}\big(\widehat{U}_t^{Y_k\mid X>\VaR}-\a\big)\\
		&= \1_{\big\{\widehat{U}_t^{X}> \b,\ \widehat{U}_t^{Y_k\mid X>\VaR}>\a\big\}}\frac{\widehat{U}_t^{Y_k\mid X>\VaR}-\a}{1-\a} \\
		&= \widetilde{H}_{kt,\a|\b}.
	\end{align*}
	Here, the second  and forth equalities use the definitions of the VaR, CoVaR and of the PITs in \eqref{eqn:PITdef}, and the third equality that $F(F^{-1}(z)) = z$ if $F$ is continuous and strictly increasing around $F^{-1}(z)$.
	The fifth equality follows from computing the integral over the indicator function.
\end{proof}

\begin{proof}[\textbf{Proof of Proposition~\ref{prop:CumViolProcess}:}]
Before proceeding, we collect two useful results.
First, $U_t^{Y_k\mid X>\VaR}\mid\{ U_t^X>\b,\mathcal{F}_{t-1}\}\sim\mathcal{U}[0,1]$, because for $x\in(0,1)$,
\begin{align}
	\P\big\{U_t^{Y_k\mid X>\VaR}\leq x\mid U_t^X>\b,\mathcal{F}_{t-1}\big\} &= \P\big\{F_{Y_{kt}\mid X_t>\VaR_{t,\b},\mathcal{F}_{t-1}}(Y_{kt})\leq x\mid X_t>\VaR_{t,\b},\mathcal{F}_{t-1}\big\} \notag \\
	&= \P\big\{Y_{kt}\leq F_{Y_{kt}\mid X_t>\VaR_{t,\b},\mathcal{F}_{t-1}}^{-1}(x)\mid X_t>\VaR_{t,\b},\mathcal{F}_{t-1}\big\} \notag \\
	&= F_{Y_{kt}\mid X_t>\VaR_{t,\b},\mathcal{F}_{t-1}}\big(F_{Y_{kt}\mid X_t>\VaR_{t,\b},\mathcal{F}_{t-1}}^{-1}(x)\big) \notag \\
	&=x. \label{eq:U_cond_Unif}
\end{align}
The statement~\eqref{eq:U_cond_Unif} extends to all $x\in[0,1]$ as $U_t^{Y_k\mid X>\VaR} \in [0,1]$ almost surely by definition.

Second, for a uniform random variable $U\sim\mathcal{U}[0,1]$ and $x\in[0,1]$, using the law of total probability,
	\begin{align}
		&\P\Big\{\1_{\{U>\a\}}\frac{U-\a}{1-\a}\leq x\Big\} \notag\\
		&= \P\Big\{\1_{\{U>\a\}}\frac{U-\a}{1-\a}\leq x \mid U>\a\Big\} \P\{U>\a\} 
		+ \P\Big\{\1_{\{U>\a\}}\frac{U-\a}{1-\a}\leq x \mid U\leq\a\Big\}\P\{U\leq\a\} \notag\\
		&= \P\Big\{\frac{U-\a}{1-\a}\leq x \mid U>\a\Big\}  \P\{U>\a\}  +  \P\Big\{0\leq x \mid U\leq\a\Big\}\a\notag\\
		&= \P\Big\{U\leq \a+x(1-\a) \mid U>\a\Big\}  \P\{U>\a\}  + \a \notag\\
		&= \P\Big\{\a < U\leq \a+x(1-\a) \Big\} + \a \notag\\
		&= x(1-\a) + \a. \label{eq:Unif trafo}
	\end{align}

Similarly as in the proof of Proposition~\ref{prop:CoVaR null}, we establish Proposition~\ref{prop:CumViolProcess} in four steps.
First, we show that $I_t\overset{d}{=}\1_{\{U_{1t}>\b\}}$. 
This, however, follows as in the proof of Proposition~\ref{prop:CoVaR null}.

Second, we establish that
\[
H_{kt,\a|\b}\overset{d}{=}\1_{\{U_{1t}>\b,\ U_{2t}>\a\}}\frac{U_{2t}-\a}{1-\a}.
\]
Similar to the proof of Proposition~\ref{rem:CumViolSeq}, but here under the stronger assumption of correctly specified forecasts, $\widehat{\VaR}_{t,\b}=F^{-1}_{X_t\mid\mathcal{F}_{t-1}}(\b)$ and $\widehat{\CoVaR}_{kt,\a|\b} = F^{-1}_{Y_{kt}\mid X_t>\VaR_{t,\b},\mathcal{F}_{t-1}}(\gamma)$ for any $\gamma \ge \alpha$ under $H_{0}^{\CoES}$,  we obtain that
\begin{align}
	H_{kt,\a|\b} 
	&= 	\frac{\1_{\{X_t> \widehat{\VaR}_{t,\beta}\}}}{1-\a}\int_{\a}^{1} \1_{\{Y_{kt}> \widehat{\CoVaR}_{kt,\gamma|\b}\}} \D \gamma \nonumber \\
	&= 	\frac{\1_{\{X_t> {\VaR}_{t,\beta}\}}}{1-\a}\int_{\a}^{1} \1_{\{Y_{kt}> {\CoVaR}_{kt,\gamma|\b}\}} \D \gamma  \nonumber \\
	&= \frac{\1_{\big\{X_t > {F}_{X_t\mid\mathcal{F}_{t-1}}^{-1}(\b)\big\}}}{1-\a}
	\int_{\a}^{1} \1_{ \big\{ Y_{kt} > {F}^{-1}_{ Y_{kt} \mid X_t > \VaR_{t,\b}, \mathcal{F}_{t-1} }(\gamma) \big\} }  \D \gamma  \nonumber 	\\
	&= \frac{\1_{\big\{ {F}_{X_t\mid\mathcal{F}_{t-1}}(X_t) > \b  \big\}}}{1-\a}	\int_{\a}^{1}\1_{\big\{ {F}_{Y_{kt}\mid X_t>\VaR_{t,\beta},\mathcal{F}_{t-1}}(Y_{kt})> \gamma \big\}} \D \gamma \nonumber \\
	&= \frac{\1_{\big\{{U}_t^{X}> \b\big\}}}{1-\a} \int_{\a}^{1}\1_{\big\{{U}_t^{Y_k\mid X>\VaR}> \gamma\big\}} \D \gamma \nonumber \\
	&= \1_{\big\{{U}_t^{X}> \b,\ {U}_t^{Y_k\mid X>\VaR}>\a\big\}}\frac{{U}_t^{Y_k\mid X>\VaR}-\a}{1-\a} \label{eqn:Hkt_derivation}.
\end{align}
Hence, we may write for all $x \in [0,1]$,
\begin{align}
	\P&\big\{H_{kt,\a|\b}\leq x \mid\mathcal{F}_{t-1}\big\}\notag\\
	&= \P\big\{H_{kt,\a|\b}\leq x \mid U_t^X>\beta,\mathcal{F}_{t-1}\big\}\P\big\{U_t^X>\beta\mid\mathcal{F}_{t-1}\big\}\notag\\
	&\hspace{2cm} + \P\big\{H_{kt,\a|\b}\leq x \mid U_t^X\leq\beta,\mathcal{F}_{t-1}\big\}\P\big\{U_t^X\leq\beta\mid\mathcal{F}_{t-1}\big\}\notag\\
	&= \P\bigg\{ \1_{\big\{U_t^{Y_k\mid X>\VaR}>\a\big\}}\frac{U_t^{Y_k\mid X>\VaR}-\a}{1-\a}\leq x\,\Big\vert\, U_t^X>\beta,\mathcal{F}_{t-1}\bigg\}(1-\b)\notag\\
	&\hspace{2cm}+ 1\cdot\b\notag\\
	& = \big[x(1-\a) + \a\big](1-\b)+\b,\label{eq:h1}
\end{align}
where we have exploited in the final step that $U_t^{Y_k\mid X>\VaR}\mid\{ U_t^X>\b,\mathcal{F}_{t-1}\}\sim\mathcal{U}[0,1]$ from \eqref{eq:U_cond_Unif} and the CDF from \eqref{eq:Unif trafo}.
By the LIE, we obtain from \eqref{eq:h1} that
\begin{equation*}
	\P\big\{H_{kt,\a|\b}\leq x\big\}=\big[x(1-\a) + \a\big](1-\b)+\b=H(x),
\end{equation*}
as claimed.

Third, we show that
\[
	(I_t, H_{kt,\a|\b})^\prime\overset{d}{=}\Big(\1_{\{U_{1t}>\b\}}, \1_{\{U_{1t}>\b,\ U_{2t}>\a\}}\frac{U_{2t}-\a}{1-\a}\Big)^\prime.
\]
To do so, we prove that for $x\in[0,1]$,
	\begin{align*}
		\P\big\{H_{kt,\a|\b}\leq x\mid I_t=0\big\} & = 1, \\
		\P\big\{H_{kt,\a|\b}\leq x\mid I_t=1\big\} & = x(1-\alpha) + \alpha.
	\end{align*}
	This fully determines the dependence structure between $H_{kt,\a|\b}$ and $I_t$, because for $y\in\{0,1\}$, by the law of total probability,
\begin{multline*}
	\P\big\{H_{kt,\a|\b}\leq x,\ I_t=y\} = \P\big\{H_{kt,\a|\b}\leq x\mid I_t=0\big\}\P\big\{I_t=0\big\}\\
	+ \P\big\{H_{kt,\a|\b}\leq x\mid I_t=1\big\}\P\big\{I_t=1\big\},
\end{multline*}
where the marginal probabilities $\P\big\{I_t=0\big\}$ and $\P\big\{I_t=1\big\}$ are already known from the first step of this proof.
We follow similar steps as in \eqref{eq:h1} and use \eqref{eqn:Hkt_derivation} to deduce that
\begin{align*}
	\P\big\{H_{kt,\a|\b}\leq x\mid I_t=0\big\} & = 1,\\
	\P\big\{H_{kt,\a|\b}\leq x\mid I_t=1\big\} & = \P\big\{{H}_{kt,\a|\b}\leq x\mid U_t^X>\beta\big\}\\
	&= \E\Bigg[\P\bigg\{ \1_{\big\{U_t^{Y_k\mid X>\VaR}>\a\big\}}\frac{U_t^{Y_k\mid X>\VaR}-\a}{1-\a}\leq x\,\Big\vert\, U_t^X>\beta,\mathcal{F}_{t-1}\bigg\}\Bigg]\\
	&=x(1-\alpha) + \alpha.
\end{align*}
This is the same conditional distribution as that of $\1_{\{U_{1t}>\b,\ U_{2t}>\a\}}(U_{2t}-\a)/(1-\a)\mid\1_{\{U_{1t}>\b\}}$, as is easy to verify.

In the fourth step, we show IIDness of $(I_t, H_{kt,\a|\b})^\prime$ under correct specification.
To do so, put $H_t:=H_{kt,\a|\b}$ for brevity.
We start by showing independence of $I_s$ and $H_{t}$ for $s\neq t$.
For $s>t$, we deduce that
\begin{align}
	\P\big\{I_s=x,\ H_{t}\leq y\mid\mathcal{F}_{s-1}\big\}&=\1_{\{H_{t}\leq y\}}\P\big\{I_s= x\mid\mathcal{F}_{s-1}\big\}\notag\\
	&= \1_{\{H_{t}\leq y\}}\big[\b\1_{\{x=0\}} + (1-\b)\1_{\{x=1\}}\big],\label{eq:CCDF H}
\end{align}
where the first step follows from $\mathcal{F}_{s-1}$-measurability of $H_{t}$, and the second step from \eqref{eq:(p.39)}.
Hence, 
\begin{align}
	\P\big\{I_s=x,\ H_{t}\leq y\big\} \notag
	&= \E\Big[\P\big\{I_s=x,\ H_{t}\leq y\mid\mathcal{F}_{s-1}\big\}\Big]\notag\\
	&= \E\Big[\1_{\{H_{t}\leq y\}}\big\{\b\1_{\{x=0\}} + (1-\b)\1_{\{x=1\}}\big\}\Big]\notag\\
	&= \P\big\{H_{t}\leq y\big\}\big[\b\1_{\{x=0\}} + (1-\b)\1_{\{x=1\}}\big]\notag\\
	&= \P\big\{H_{t}\leq y\big\}\P\big\{I_s=x\big\},\label{eq:prob prod}
\end{align}
where we used the LIE in the first step, \eqref{eq:CCDF H} in the second step, and \eqref{eq:(p.39)} in the final step.
The case $s<t$ can be dealt with similarly, such that \eqref{eq:prob prod} holds for all $s\neq t$.

Adopting similar arguments, we may also deduce that
\begin{multline*}
	\P\big\{I_{s_1}= x_1,\ldots,I_{s_{\ell}}= x_{\ell},\ H_{t_1}\leq y_1,\ldots,H_{t_m}\leq y_{m}\big\}\\
	=\P\big\{I_{s_1}= x_1\big\}\cdot\ldots\cdot \P\big\{I_{s_{\ell}}= x_{\ell}\big\}\cdot \P\big\{H_{t_1}\leq y_1\big\}\cdot\ldots\cdot \P\big\{H_{t_m}\leq y_m\big\}
\end{multline*}
for any indices $s_1<\ldots<s_{\ell}$ and $t_1<\ldots<t_m$ ($\ell,m\in\mathbb{N}$) with $s_i\neq t_j$ ($i=1,\ldots,\ell$ and $j=1,\ldots,m$), as required for full independence.
\end{proof}

\end{document}